\documentclass[11pt]{amsart}

\usepackage{adjustbox,array}
\usepackage{fullpage}
\usepackage{amsmath, amssymb, amsthm}
\usepackage{booktabs}
\usepackage{comment}
\usepackage[foot]{amsaddr}
\usepackage{algorithm,algpseudocode}
\usepackage{graphicx,subcaption}
\usepackage{multirow}
\usepackage{tikz}
\usetikzlibrary{cd}
\usepackage{url}
\usepackage[normalem]{ulem}

\usepackage{siunitx,etoolbox}

\usepackage{color}
\usepackage[colorlinks=true,linkcolor=red,citecolor=blue,urlcolor=cyan]{hyperref}
\usepackage{cleveref}

\usepackage{svg}
\usepackage{xparse}

\newcolumntype{R}[2]{%
    >{\adjustbox{angle=#1,lap=\width-(#2)}\bgroup}%
    l%
    <{\egroup}%
}
\newcommand*\rot{\multicolumn{1}{R{45}{1em}}}

\Crefname{equation}{}{}
\Crefname{app}{Appendix}{Appendices}
\Crefname{app}{Appendix}{Appendices}

\newcommand{\Ex}{\mathbb{E}}

\newtheorem{theorem}{Theorem}

\newtheorem{corollary}[theorem]{Corollary}

\newtheorem{lemma}[theorem]{Lemma}

\newtheorem{proposition}[theorem]{Proposition}
\newtheorem{remark}[theorem]{Remark}

\numberwithin{equation}{section}
\numberwithin{theorem}{section}

\DeclareMathOperator{\prob}{Prob}

\title[Escape times for subgraph detection and graph partitioning]{An escape time formulation for subgraph detection and partitioning of directed graphs}\thanks{The authors are grateful to Matthias Kurzke and Jonathan Weare for helpful discussions throughout the crafting of this result. Z.M.B.\ and P.J.M.\ were supported by the ARO under MURI award W911NF-18-1-0244. J.L.M.\ acknowledges support from the NSF through grant DMS-1909035. P.J.M.\ was also supported by the NSF through grant BCS-2140024. B.\ Osting acknowledges support from NSF DMS-1752202. Z.M.B.\ was supported by DMS-2137511.}

\author[Z.~M.~Boyd]{Zachary~M.~Boyd} \address{Department of Mathematics, Brigham Young University, Provo, UT 84602, USA} \email{zachboyd@byu.edu}

\author[N.~Fraiman]{Nicolas Fraiman} \address{Department of Statistics and Operations Research, University of North Carolina at Chapel Hill, Chapel Hill, NC 27599, USA} \email{fraiman@email.unc.edu}

\author[J.~L.~Marzuola]{Jeremy~L.~Marzuola} \address{Department of Mathematics, University of North Carolina at Chapel Hill, Chapel Hill, NC 27599, USA} \email{marzuola@math.unc.edu}

\author[P.~.J.~Mucha]{Peter~J.~Mucha} \address{Department of Mathematics, Dartmouth College, Hanover, NH 03755, USA} \email{peter.j.mucha@dartmouth.edu}

\author[B.~Osting]{Braxton Osting} 
\address{Department of Mathematics, University of Utah, Salt Lake City, UT 84112, USA} 
\email{osting@math.utah.edu}

\begin{document}

\begin{abstract} 
We provide a rearrangement based algorithm for fast detection of subgraphs of $k$ vertices with long escape times for directed or undirected networks. Complementing other notions of densest subgraphs and graph cuts, our method is based on the mean hitting time required for a random walker to leave a designated set and hit the complement.  We provide a new 
relaxation of this notion of hitting time on a given subgraph and use that relaxation to construct a fast subgraph detection algorithm and a generalization to $K$-partitioning schemes.  Using a modification of the subgraph detector on each component, we propose a graph partitioner that identifies regions where random walks live for comparably large times. Importantly, our method implicitly respects the directed nature of the data for directed graphs while also being applicable to undirected graphs. We apply the partitioning method for community detection to a large class of model and real-world data sets.  
\end{abstract}

\maketitle


\section{Introduction}


Subgraph detection and graph partitioning are fundamental problems in network analysis, each typically framed in terms of identifying a group or groups of vertices of the graph so that the vertices in a shared group are well connected or ``similar'' to each other in their connection patterns while the vertices in different groups (or the complement group) are ``dissimilar''. The specific notion of connectedness or similarity is a modeling choice, but one often assumes that edges connect similar vertices, so that in general the detected subgraph is dense and the ``communities'' identified in graph partitioning are very often more connected within groups than between groups (``assortative communities'').

The identification of subgraphs with particular properties is a long-standing pursuit of network analysis with various applications. Dense subgraphs as assortative communities might represent coordinating regions of interest in the brain \cite{Meunier_2009,Bassett_2011} or social cliques in a social network \cite{Moody_2001}. In biology, subgraph detection plays a role in discovering DNA motifs and in gene annotation \cite{fratkin2006motifcut}. In cybersecurity, dense subgraphs might represent anomalous patterns to be highlighted and investigated (e.g., \cite{Yan_2021}).  See \cite{ma2020efficient} for a recent survey and a discussion of alternative computational methods.  As noted there, some of the existing algorithms apply to directed graphs, but most do not.

In the corresponding computer science literature, much of the focus has been on approximation algorithms since the dense $k$-subgraph is NP-hard to solve exactly (a fact easily seen by a reduction from the $k$-clique problem). An algorithm that on any input $(G, k)$ returns a subgraph of order $k$ (that is, $k$ vertices or ``nodes''; note, we will sometimes refer to the ``size'' of a graph or subgraph to be the number of vertices, not the number of edges) with average degree within a factor of at most $n^{1/3-\delta}$ from the optimum solution, where $n$ is the order of graph $G$ and $\delta\approx 1/60$ was proposed in \cite{feige2001dense}. This approximation ratio was the best known for almost a decade until a log-density based approach yielded $n^{1/4+\varepsilon}$
for any $\varepsilon > 0$ \cite{bhaskara2010}. This remains the state-of-the-art approximation algorithm. On the negative side it has been shown \cite{manurangsi2017}, assuming the exponential time hypothesis, that there is no polynomial-time algorithm that approximates to within an $n^{1/(\log\log n)^c}$ factor of the optimum. Variations of the problem where the target subgraph has size at most $k$ or at least $k$ have also been considered \cite{andersen2009}.

Depending on the application of interest, one might seek one or more dense subgraphs within the larger network, a collection of subgraphs to partition the network (i.e., assign a community label to each node), or a set of potentially overlapping subgraphs (see, e.g., \cite{Wilson_2014}). While the literature on ``community detection'' is enormous (see, e.g.,  \cite{fortunato_community_2010,fortunato_community_2016,fortunato_community_2022,porter_communities_2009,shai_case_2017} as reviews), a number of common thematic choices have emerged. Many variants of the graph partitioning problem can be formalized as a (possibly constrained) optimization problem. One popular choice minimizes the total weight of the cut edges while making the components roughly equal in size \cite{shi2000normalized}. Another common choice maximizes the total within-community weight relative to that expected at random in some model \cite{newman_finding_2004}. Other proposed objective functions include ratio cut weight \cite{bresson2013adaptive}, and approximate ``surprise'' (improbability) under a cumulative hypergeometric distribution \cite{Traag_Aldecoa_Delvenne_2015}. 
However, most of these objectives are NP-hard to optimize, leading to the development of a variety of heuristic methods for approximate partitioning (see the reviews cited above for many different approaches). Some of the methods that have been studied are based on the Fielder eigenvector \cite{fiedler1973algebraic}, multicommunity flows \cite{leighton1989approximate}, semidefinite programming \cite{arora2004expander,arora2008geometry,arora2009expander}, expander flows \cite{arora2010logn}, single commodity flows \cite{khandekar2009graph}, or Dirichlet partitions  \cite{osting2014minimal,osting2017consistency,wang2019diffusion}.

Whichever choice is made for the objective and heuristic, the identified communities can be used to describe the mesoscale structure of the graph and can be important in a variety of applications (see, e.g., the case studies considered in \cite{shai_case_2017}). Subgraphs and communities can also be important inputs to solving problems like graph traversal, finding paths, trees, and flows; while partitioning large networks is often an important sub-problem for complexity reduction or parallel processing in problems such as graph eigenvalue computations \cite{BDR13}, breadth-first search \cite{BM13}, triangle listing \cite{CC11}, PageRank \cite{SW13} and Personalized PageRank~\cite{andersen2007algorithms}.

In the present work, we consider a different formulation of the subgraph detection problem, wherein we aim to identify a subgraph with a long mean exit time---that is, the expected time for a random walker to escape the subgraph and hit its complement. Importantly, this formulation inherently respects the possibly directed nature of the edges. This formulation is distinct from either maximizing the total or average edge weight in a dense subgraph and minimizing the edge cut (as a count or suitably normalized) that is necessary to separate a subgraph from its complement. Furthermore, explicitly optimizing for the mean exit time to identify subgraphs may in some applications be preferred as a more natural quantity of interest. For example, in studying the spread of information or a disease on a network, working in terms of exit times is more immediately dynamically relevant than structural measurements of subgraph densities or cuts. Similarly, the development of respondent-driven sampling in the social survey context (see, e.g., \cite{Mouw_2012,Verdery_2015}) is primarily motivated by there being subpopulations that are difficult to reach (so we expect they often also have high exit times on the directed network with edges reversed). We thus argue that the identification of subgraphs with large exit times is at least as interesting---and typically related to---those subgraphs with large density and or small cut. Indeed, random walker diffusion on a network and assortative communities are directly related in that the modularity quality function used in many community detection algorithms can be recovered as a low-order truncation of a ``Markov stability'' auto-correlation measurement of random walks staying in communities \cite{Lambiotte_Delvenne_Barahona_2014}. However, the directed nature of the edges is fully respected in our escape time formulation of subgraph detection presented here (cf.\ random walkers moving either forward or backward along edges in the Markov stability calculation \cite{Mucha_Richardson_Macon_Porter_Onnela_2010} that rederives modularity for a directed network \cite{Leicht_Newman_2008}). 

From an optimization point of view, the method presented here can be viewed as a rearrangemnet method or a Merriman-Bence-Osher (MBO) scheme \cite{MBO1993} as applied to Poisson solves on a graph.  Convergence of MBO schemes is an active area of research in a variety of other scenarios: see \cite{chambolle2006convergence,ishii2005optimal} in the case of continuum mean curvature flows, \cite{budd2020graph,van_Gennip_2014} in a graph Allen-Cahn type problem, and \cite{jacobs2018auction} for a volume constrained MBO scheme on undirected networks. 
Similarly, proving convergence rates for our algorithm by determining quantitative bounds on the number of interior iterations required for a given $\epsilon$ is an important question for the numerical method and its applications to large data sets. Importantly, the method for subgraph detection that we develop and explore, and then extend to a partitioner, is inherently capable of working on directed graphs without any modification. 
Also, searching for related graph problems where this type of rearrangement algorithm for optimization can be applied will be an important endeavor.

\subsection{A New Formulation in Graphs}

Let $G = (V,E)$ be a (strongly) connected graph (undirected or directed; we use the term ``graph'' throughout to include graphs that are possibly directed), with adjacency matrix $A$ with element $A_{ij}$ indicating presence/absence (and possible weight) of an edge from $i$ to $j$. We define the (out-)degree matrix $D$ to be diagonal with values $D_{ii}=\sum_j A_{ij}$. For weighted edges in $A$ this weighted degree is typically referred to as ``strength'' but we will continue to use the word ``degree'' throughout to be this weighted quantity. Consider the discrete time Markov chain $M_n$ for the random walk described by the (row stochastic) \emph{probability transition matrix}, $P := D^{-1} A$. The \emph{exit time from $S\subset V$} is the stopping time $T_S = \inf\{n\geq 0: M_n\in S^c\}$.
The \emph{mean exit time from $S$ of a node $i$} is defined by $\Ex_i T_S$ (where $\Ex_i$ is the expectation if the walker starts at node $i$) and is given by $v_i$, where $v$ is the solution to the system of equations 
\begin{subequations}
\label{ht_def}
\begin{align}
\label{ht_defa}
(I-P)_{SS} v_S &= 1_S \\
v_{S^c}  &= 0\,, 
\end{align}
\end{subequations}
where the subscript $S$ represents restriction of a vector or matrix to the indices in $S$. 
The \emph{average mean escape time (MET) from $S$} is then
\begin{equation} \label{e:MET}
\tau (S) = \frac{1}{|V|} \sum_{i \in V} v_{i},
\end{equation} 
representing the mean exit time from $S$ of a node chosen uniformly at random in the graph (noting that $v_i=0$ for $i\in S^c$).
We are interested in finding vertex sets (of fixed size) having large MET, as these correspond to sets that a random walker would remain in for a long time. Thus, for fixed $k\in \mathbb N$, we consider the \emph{subgraph detection problem}, 
\begin{equation}
 \label{e:subgraphDetection}
 \max_{\substack{S\subset V \\ |S| =k}} \tau(S).
 \end{equation}
%
Multiplying~\eqref{ht_defa} on the left by $D$, we obtain the equivalent system, 
\begin{subequations}
\label{e:Poisson}
\begin{align}
& L v = d 
\textrm{ on } S , \\
& v = 0 \textrm{ on } S^c\,,
\end{align}
\end{subequations}
where $L = D-A$ is the (unnormalized, out-degree) graph Laplacian, and $d = D 1$ is the out-degree vector.  We denote the solution to \eqref{e:Poisson} by $v = v(S)$. 
For $\varepsilon> 0$, we will also consider the approximation to \eqref{e:Poisson}, 
\begin{equation}
\label{e:relax}
\left[ L + \varepsilon^{-1}  (1-\phi) \right] u =  d 
\end{equation}
where $\phi$ is a vector and action by $(1-\phi)$ on the left is interpreted as multiplication by the diagonal matrix $I - {\rm diag} (\phi)$.  We denote the solution $u = u_\varepsilon$. 
Formally, for $\phi = \chi_S$, the characteristic function of $S$,  as $\varepsilon \to 0$, the vector $u_\varepsilon \to v_S$ where $v_S$ satisfies \eqref{e:Poisson}.
We can also define an associated approximate MET
\begin{equation}
\label{e:l1energy}
E_\varepsilon (\phi) := \frac{1}{|V|} \| u_\varepsilon \|_{
\ell^1(V)} =   \frac{1}{|V|}  \left\|  \left[  L + \varepsilon^{-1}  (1-\phi) \right]^{-1} d \right\|_{\ell^1 (V)},
\end{equation}
where as $\varepsilon \to 0$, we have that $E_\varepsilon(\chi_S) \to \frac{1}{|V|}  \| v_S \|_{\ell^1(V)} = \tau(S)$.  We then arrive at the following \emph{relaxed subgraph detection problem}
\begin{equation}
\label{exp:l1_opt}
\max_{\substack{0 \leq \phi \leq 1 \\ \langle \phi, 1 \rangle = k}} E_\epsilon (\phi),
\end{equation}
which we solve and study in this paper.
For small $\varepsilon>0$, we will study the relationship between the subgraph detection problem \eqref{e:subgraphDetection} and its relaxation
\eqref{exp:l1_opt}.

We are also interested in finding node partitions with high MET in the following sense: Given a vertex subset $S \subset V$, a random walker that starts in $S$ should have difficulty escaping to $S^c$ and a random walker that starts in $S^c$ should have difficulty escaping to $S$. This leads to the problem $\max_{V = S \amalg S^c} \,\tau(S) + \tau(S^c)$. 
More generally, for a vertex partition, $V = \amalg_{\ell \in [K]} S_\ell$, we can consider 
\begin{equation} \label{e:minEscape}
\max_{V = \amalg_{\ell \in [K]} S_\ell} \ \sum_{\ell \in [K]} \ \tau(S_\ell). 
\end{equation}
The solution embodies the idea that in a good partition a random walker will transition between partition components very infrequently. 
An approximation to \eqref{e:minEscape} is
\begin{equation}
\label{e:Opt}
\max_{V = \amalg_{\ell \in [K]} S_\ell} \ \sum_{\ell \in [K]} \ E_\varepsilon (\chi_{S_\ell}). 
\end{equation}
We can make an additional approximation by relaxing the constraint set. 
Define the admissible class 
$$
\mathcal A_K = \left\{ \{\phi_\ell\}_{\ell\in [K]} \colon 
\phi_\ell \in \mathbb R^{|V|}_{+}
\text{ and } 
\sum_{\ell \in [K] } \phi_\ell  = 1 \right\}.
$$
Observe that the collection of indicator functions for any $K$-partition of the vertices  is a member of $\mathcal A_K$. Furthermore, we can see that $\mathcal A_K \cong (\Delta_K)^{|V|}$, where $\Delta_K$ is the unit simplex in $K$ dimensions. Thus, the extremal points of $\mathcal A_K$ are precisely the collection of indicator functions for a $K$-partition of the vertices. 
For $\varepsilon >0$, a modified relaxed version of the graph partitioning problem \eqref{e:minEscape} can be formulated as
\begin{equation}
   \label{e:minEscapeRelax_alt}
\min_{ \{\phi_\ell\}_{\ell \in [K]} \in \mathcal A_K } \tilde{E}_{\epsilon} \left( \{\phi_\ell\}_{\ell \in [K]} \right), \quad \textrm{where} \quad  \tilde{E}_{\epsilon} \left( \{\phi_\ell\}_{\ell \in [K]} \right) = \sum_{i = 1}^K [1+ \epsilon |V| E_\epsilon(\phi_i)]^{-1}.
\end{equation}
For small $\varepsilon>0$, we will study the relationship between the graph partitioning problem  \eqref{e:minEscape} and its relaxation
\eqref{e:minEscapeRelax_alt}.  An important feature of \eqref{e:minEscapeRelax_alt} is that it can be optimized using fast rearrangement methods that effectively introduces a volume normalization for the partition sets, while optimization of \eqref{e:minEscape} resulted in favoring one partition being full volume.  We will discuss this further in Section \ref{sec:gp} below.

\subsection{Outline of the Paper}
In \Cref{s:Analysis}, we lay the analytic foundation for rearrangement methods for both the subgraph detection and partitioning problems.  
We prove the convergence of the methods to local optimizers of our energy functionals in both cases and establish the fact that our fast numerical methods increase the energy.  
To begin, we establish properties of the gradient and Hessian of the functionals $E_\epsilon (\phi)$ for vectors $0 \leq \phi \leq 1$.  Then, using those properties, we introduce rearrangement methods for finding optimizers and prove that our optimization schemes reduce the energy. 
Then, we discuss how to adapt these results to the partitioning problem.  
Lastly, we demonstrate how one can easily add a semi-supervised component to our algorithm.   

In \Cref{s:NumRes}, we apply our methods to a variety of model graphs, as well as some empirical data sets to assess their performance. In the subgraph setting, we consider how well we do detecting communities in a family of model graphs related to stochastic block models, made up of a number of random Erd\H{o}s-R\'enyi (ER) communities of various sizes and on various scales.  The model graphs are designed such that the overall degree distribution is relatively similar throughout.  We demonstrate community detectability and algorithm efficacy thresholds by varying a number of parameters in the graph models.  We also consider directed graph models of cycles connected to Erd\H{o}s-R{\'e}nyi graphs, on which our methods perform quite well.  For the partitioners, we also consider related performance studies over our model graph families, as well as on a large variety of clustering data sets.  

We conclude in \Cref{s:disc} with a discussion including possible future directions and applications of these methods.  

\section{Analysis of our proposed methods}
\label{s:Analysis}

In this section, we first analyze the relaxed subgraph detection problem \cref{exp:l1_opt} and the relaxed graph partitioning problem \Cref{e:minEscapeRelax_alt}. Then, we propose and analyse computaitonal methods for the problems. As noted above, we assume throughout that the graph is (strongly) connected.
\subsection{Analysis of the relaxed subgraph detection problem and the relaxed graph partitioning problem}


For fixed $\epsilon > 0$ and 
$\phi \in [0,1]^{|V|}$, denote the operator on the RHS of  \Cref{e:relax} by
$L_\phi := D-A +  \frac{1}{\epsilon} (1-\phi)$. 

\begin{lemma}[Discrete maximum principle]
\label{lem:max}  
Given the regularized operator $L_\phi$ and a vector $f > 0$, we have $(L_{\phi}^{-1} f)_v > 0$ for all $v \in V$. Without strong connectivity, this result still holds (with $>$ replaced by $\ge$) as long as there are no leaf nodes.
\end{lemma}

\begin{proof}
Writing $L_\phi = \left( D + \frac{1}{\epsilon} (1-\phi) \right)  - A$, we observe that 
\begin{align*}
L_\phi^{-1} & = \left(   \left( D + \frac{1}{\epsilon} (1-\phi) \right) \left(  I   -   \left( D + \frac{1}{\epsilon} (1-\phi) \right)^{-1} A \right) \right)^{-1}  \\
& = \left(  I   -   \left( D + \frac{1}{\epsilon} (1-\phi) \right)^{-1} A \right)^{-1}   \left( D + \frac{1}{\epsilon} (1-\phi) \right)^{-1} \\
& = \sum_{n=0}^\infty  \left[ \left( D + \frac{1}{\epsilon} (1-\phi) \right)^{-1} A \right]^n \left( D + \frac{1}{\epsilon} (1-\phi) \right)^{-1}.
\end{align*}
Since all entries in the corresponding matrices are positive (by strong connectivity), the result holds.
\end{proof}

For simplicity, in the following we consider simply setting the potential $X : = \epsilon^{-1}  ( 1- \phi)$ and we use $X$ and ${\rm diag}\ X$ interchangeably for graph Schr\"odinger operators of the form $L_X := D-A + X$ and solutions of the Poisson equation $L_X u = d$. We can then consider the related energy functional
\begin{equation}
\label{e:l1energy_alt}
E (X) : =  \left\|  \left[  L + X \right]^{-1} d \right\|_{\ell^1 (V)} = \| u\|_{
\ell^1(V)} .
\end{equation}

\begin{lemma}
\label{diff:lem}
The gradient of $E(X)$ with respect to $X$ is given by
\begin{equation}
  \label{Jgrad}
   \nabla E = - u \odot v 
\end{equation}
where $\odot$ denotes the Hadamard product and 
\begin{equation}
    \label{uvdef}
    u = (L+X)^{-1} d,  \ \  v = (L+X)^{-T} e.
\end{equation}  
Here $e$ is the all-ones vector. The Hessian of $E(X)$ with respect to $X$ is then given by
\begin{equation}
\label{Jhess}
    H = \nabla^2 E =  (L + X)^{-1} \odot W + (L + X)^{-T} \odot W^T
\end{equation}
where
\begin{equation*}
W := u \otimes v 
\end{equation*}
where $\otimes$ is the Kronecker (or outer) product.
\end{lemma}

\begin{proof}
Write $e_j$ as the indicator vector for the $j$th entry. First, differentiating~\cref{uvdef} with respect  to $X_j$, we compute 
$$
(L + X) \frac{\partial u}{ \partial X_j} = - e_j \odot u
\qquad \implies \qquad 
\frac{\partial u}{ \partial X_j} = -  \langle e_j,  u \rangle (L + X)^{-1} e_j. 
$$
Taking the second derivative, we obtain 
\begin{align*}
(L + X) \frac{\partial^2 u}{ \partial X_j \partial X_k} 
&= - e_j \left\langle e_j,   \frac{\partial u}{ \partial X_k} \right\rangle  
- e_k \left\langle e_k ,  \frac{\partial u}{ \partial X_j} \right\rangle \\ 
&=  e_j \langle e_k,  u \rangle  \left \langle e_j, (L + X)^{-1} e_k \right\rangle
+  e_k \langle e_j,  u \rangle  \left \langle e_k, (L + X)^{-1} e_j \right\rangle, 
\end{align*}
which implies that 
$$
\frac{\partial^2 u}{ \partial X_j \partial X_k} =
\left \langle e_j, (L + X)^{-1} e_k \right\rangle  \langle e_k,  u \rangle (L+X)^{-1}e_j + \left \langle e_k, (L + X)^{-1} e_j \right\rangle 
\langle e_j,  u \rangle (L+X)^{-1}e_k. 
$$

By the maximum principle (Lemma~\ref{lem:max}), $u$ is positive and we can write 
$ E(X) = \| u \|_1 =  \langle e, u \rangle$. Thus, the gradient is
\begin{align*}
    \frac{\partial E}{ \partial X_j} 
&= \left \langle e, \frac{\partial u}{ \partial X_j} \right \rangle \\
& = - \langle (L+X)^{-T} e, e_j \rangle \langle u, e_j \rangle,
\end{align*}
or in other words
\[
\nabla_X E = u \odot v
\]
for $u$ and $v$ as in \eqref{uvdef}.

For the Hessian, we have
\begin{align*}
&\frac{\partial^2 E}{ \partial X_j \partial X_k} 
= \left \langle e, \frac{\partial^2 u}{ \partial X_j \partial X_k} \right \rangle \\
& \hspace{.2cm} = \left \langle e_k, (L + X)^{-1} e_j \right\rangle 
\langle u,  e_j \rangle \left \langle  e_k , v\right \rangle + \left \langle e_j, (L + X)^{-1} e_k \right\rangle  \langle v, e_j \rangle  \left \langle e_k,u \right \rangle.
\end{align*}
Thus, the Hessian can be written 
$$
H = \nabla^2 E = (L + X)^{-1} \odot W + (L + X)^{-T} \odot W^T
$$
where
\begin{equation*}
W := u \otimes v.
\end{equation*}
as claimed.
\end{proof}

\begin{remark}
    If $L$ is symmetric, the above statements can be simplified greatly to give
$$
H = \nabla^2 E = (L + X)^{-1} \odot (W+W^T)
$$
where
\begin{equation*}
W + W^T := u \otimes v + v \otimes u  =  \frac{1}{2} (u + v) \otimes (u+w) - \frac{1}{2} (u - v) \otimes (u- v) . 
\end{equation*}
\end{remark}

\begin{proposition}
\label{p:Covexity}  
For $f > 0$ fixed, let $u$ satisfy $
(L+X) u = f$.
The mapping 
$X \mapsto E(X) = \| u \|_1$ is strongly convex on $\{X \geq 0, \ X \neq 0 \}$. 
\end{proposition}

\begin{proof} 

We wish to show that 
\[E(X) = e^T (L +X)^{-1} d \] 
is convex on $[0,X_{\infty}]^n$ for fixed constant $X_{\infty}$. 
Replacing $D+X$ with $X$, this is equivalent to 
\[e^T (X-A)^{-1} d\]
being convex on $\{V : d_i + X_{\infty} \ge X_i \ge d_i\}$. 
Expanding, we have 
\[e^T \left( I - X^{-1} A \right)^{-1} X^{-1} d = e^T \sum_{k=0}^{\infty} \left(X^{-1} A \right)^k X^{-1} d.\]
So it is enough to show that 
\[e^T \left(X^{-1} A \right)^k X^{-1} d\]
is convex for each $ k > 0 $. 
This is true as long as
\[ f(x) = \prod_i x_i^{-\alpha_i}\]
is convex for any $\alpha = (\alpha_1,\cdots,\alpha_n)$. 
Computing second derivatives gives
\[f_{X_iX_i}(X) = f(X) \alpha_i (\alpha_i + 1) X_i^{-2}\]
and 
\[f_{X_i X_j}(X) = f(X) \alpha_i \alpha_j  X_i^{-1} X_j^{-1}.\] 
So the Hessian of $f$ is 
\[f(X) \left[ (\alpha X^{-1})^T (\alpha X^{-1}) + \textrm{diag}(\alpha X^{-2})\right],\] which is clearly positive semi-definite, being the sum of positive semi-definite matrices.

To observe strong convexity, recognize that the $k = 0$ term contributes a term to the Hessian of the form $DX^{-2}$, which is positive definite on the domain in question.
\end{proof}

Proposition \ref{p:Covexity}  gives that $\phi \to E_\varepsilon(\phi)$ is strongly convex on $\mathbb R^{|V|}_+$, so $\{\phi_\ell\}_{\ell \in [K]}  \mapsto \mathcal{E}^\varepsilon\left( \{\phi_\ell\}_{\ell \in [K]} \right) $ is also convex on $\mathcal{A}_K$.  The following corollary is then immediate.

\begin{corollary}[Bang-bang solutions] \label{c:BangBang}
Every maximizer of \Cref{exp:l1_opt} is an extreme point of $\{ \phi \in [0,1]^{|V|} \colon \langle \phi,1\rangle =k \}$, {\it i.e.}, an indicator function for some vertex set $S \subset V$ with $|S| =k$.

\end{corollary}

Thus, in the language of control theory,  
\Cref{c:BangBang} shows that 
\Cref{exp:l1_opt}  
is a bang-bang relaxation of 
\eqref{e:subgraphDetection}
and that
\eqref{e:minEscapeRelax_alt}  
is a bang-bang relaxation of 
\eqref{e:minEscape}. 

\bigskip

\begin{corollary}
\label{Hess:cor}
Since the set of values $(x_1,\dots,x_n) \in \mathbb{R}^n_+$ with which we are concerned is convex and $E$ is $C^2$ in $X$, the resulting Hessian matrix $H$ is positive definite.  
\end{corollary}

\begin{remark}
Note that though the Hadamard product of two positive definite matrices is positive definite, Corollary \ref{Hess:cor} is not obvious from the structure of the Hessian, given that the matrix $W$ is indefinite when $u$ and $w$ are linearly independent.  As a result, this positive definiteness is strongly related to the structure of the $L+X$ matrix and its eigenvectors.  
\end{remark}


\subsection{Optimization scheme}

\subsubsection{Subgraph detector}
We solve~\Cref{exp:l1_opt} using rearrangement ideas as follows. After initializing $S$ (randomly in our experiments), we use the gradient~\Cref{Jgrad} to find the locally optimal next choice of $S$, and then iterate until convergence (typically $<10$ iterations in our experiments). More explicitly, we follow these steps:
\begin{align}
L u + \epsilon^{-1} (1- \chi_{S^0}) u  & = d, \label{eq:grad_comp1} \\
L^T v + \epsilon^{-1} (1- \chi_{S^0}) v  & = 1.
\label{eq:grad_comp}
\end{align}
The update, $S^1$, then contains those nodes $\ell$ that maximize $u_{\ell}v_{\ell}$.

\begin{algorithm}[t]
\caption{Subgraph detector}
\label{alg:subgraph}
\begin{algorithmic}
\State Input $S^0 \subset V$.
\While{$S^t \ne S^{t-1}$}
\State Solve~\Cref{eq:grad_comp1} and~\Cref{eq:grad_comp} for $u$ and $v$.
\State  Assign vertex $\ell$ to subgraph $S^1$ if $\nabla_\phi E$ is optimized. That is, solve the following sub-problem.
\begin{equation}
\max_{|S| = k}  \ \sum_{\ell \in S} u(\ell) \cdot v (\ell)  .
\label{subgraph_inner}
\end{equation}
(Note that~\Cref{subgraph_inner} is easily solved by taking the $k$ indices corresponding to the largest values of $u(\ell) \cdot v(\ell)$, breaking ties randomly if needed.)
\State Reset now, building on $ S^1 \subset V$ accordingly and repeat until $S^n = S^{n-1}$.
\EndWhile
\end{algorithmic}
\end{algorithm}

Pseudocode for this approach is given in~\Cref{alg:subgraph}, which has the following ascent guarantee:
\begin{proposition}
\label{prop:sgascent}
Every nonstationary iteration of~\cref{alg:subgraph} strictly increases the energy $E_\epsilon$. \Cref{alg:subgraph} terminates in a finite number of iterations. 
\end{proposition}

\begin{proof}
Let $S^0$ and $S^1$ be the vertex subsets for successive iterations of the method. 
Define  $W = \chi_{S^1} - \chi_{S^0}$. Assuming $W \neq 0$, by strong convexity (Theorem~\ref{p:Covexity}) and the formula for the gradient \eqref{Jgrad}, we compute
\begin{subequations}
\begin{align}
E_\epsilon (\chi_{S^1}) &> E_\epsilon (\chi_{S^0}) + \frac{1}{\epsilon} \langle W, uv \rangle \\
&= E_\epsilon (\chi_{S^0}) + \frac{1}{\epsilon} \left( \sum_{i \in S_1} u_i v_i - \sum_{i \in S_0} u_i v_i \right)   \\
& \geq E_\epsilon (\chi_{S^0}). 
\end{align}
\end{subequations}
Thus, the  energy is strictly increasing on non-stationary iterates. 
Since we assume that $V$ is a finite size vertex set and the rearrangement method increases the energy, it cannot cycle and hence must terminate in a finite number of iterations.
\end{proof}

To avoid hand-selection of $\epsilon$, we always set $\epsilon = C/\lambda_F$, where $\lambda_F$ is the Frobenius norm of the graph Laplacian and $C >1$ is typically set at $C=50$ to make sure $\epsilon$ allows communication between graph vertices.  If $C$ is chosen to take a different value below, we will highlight those cases.

\subsubsection{Graph partitioner}
\label{sec:gp}

Given the success of the energy \eqref{e:l1energy}, one might naively consider partitioning the graph by maximizing an energy of the form
\begin{equation}
\label{part:energy_bad}
(S_1, S_2, \dots, S_K) \mapsto \sum_{i = 1}^K [E_\epsilon(\chi_{S_i})].
\end{equation}
However, it can be computed that this energy does not properly constrain the volumes of each partition in a reasonable fashion and the optimizer of this nice problem merely results in putting all the vertices in a single box.

The partition energy we initially worked to minimize instead is of the form
\begin{equation}
\label{part:energy}
 (S_1, S_2, \dots, S_K) \mapsto \sum_{i = 1}^K [ |V| E_\epsilon( \chi_{S_i} )]^{-1},
\end{equation}
since the inverses penalize putting all nodes into the same partition by making the resulting empty classes highly costly.  Intuitively, this energy functional provides an effective volume normalization of the relative gradients (similar to a K-means type scheme). However, while in practice this functional appears to work reasonably well on all graph models considered here, we were unable to prove, upon analysis of the Hessian, that rearrangements based on such an algorithm are bang-bang like the subgraph detector. 

As an alternative, we instead consider
\begin{equation}
    \label{part:energyalt}
\tilde{E}_{\delta, \epsilon} (S_1, S_2, \dots, S_K) = \sum_{i = 1}^K [1+ \delta |V| E_\epsilon( \chi_{S_i})]^{-1}.
\end{equation}
Applied to functions, $0 \leq \phi_j \leq 1$, instead of indicator functions, we consider
\begin{equation}
    \label{part:energyalt_phi}
\tilde{E}_{\delta,\epsilon} (\phi_1, \phi_2, \dots, \phi_K) = \sum_{i = 1}^K [1+ \delta |V| E_\epsilon(\phi_i)]^{-1}.
\end{equation}
We then have that 
\begin{equation}
    \label{Grad:pe}
\nabla_{\phi_j} \tilde{E} = - \frac{\delta}{[1+ \delta |V| E_\epsilon(\phi_i)]^{2}} \nabla_{\phi_j} ( |V| E_\epsilon (\phi_j))
\end{equation}
making the Hessian consist of blocks of the form
\begin{align}
    \label{Hess:pe}
\nabla^2_{\phi_j} \tilde{E} & = - \frac{\delta}{[1+ \delta |V| E_\epsilon(\phi_i)]^{2}} \nabla^2_{\phi_j} (|V| E_\epsilon (\phi_j)) \\
& \hspace{.5cm} + 2  \frac{\delta^2}{[1+ \delta |V| E_\epsilon(\phi_i)]^{3}} (\nabla_{\phi_j} ( |V| E_\epsilon (\phi_j)) ) ( \nabla_{\phi_j} ( |V| E_\epsilon (\phi_j)) )^T. \notag
\end{align}
Note, this is the sum of a negative definite operator and a rank one matrix, meaning that for $\delta$ sufficiently small, the Hessian will prove that $\tilde E$ is concave with respect to each component.  In practice, we find that taking $\delta=\epsilon$ is sufficient both for having a negative definite Hessian and generating good results with respect to our rearrangement scheme.  As such, we will generically take $\delta = \epsilon$ henceforward.

Our approach to the node partitioner is largely analogous to that of the subgraph detector, with the exception that we use class-wise $\ell^1$ normalization when comparing which values of $u \cdot v$ at each node. In detail, the algorithm is presented in Algorithm \ref{alg:partitioner}.   It is a relatively straightforward exercise applying the gradient computation for $E_\epsilon (S_i)$ from Proposition \ref{p:Covexity} to prove that the energy functional \eqref{part:energy} will decrease with each iteration of our algorithm as in Proposition \ref{prop:sgascent}.

\begin{algorithm}[t]
\caption{Graph Partitioner}
\label{alg:partitioner}
\begin{algorithmic}
\State Input $\vec{S} = \{ S_1^0, \dots, S_K^0 \}$ a $K$ partition of $V$.
\While{${\vec S}^t \ne {\vec S}^{t-1}$}
\State For $j = 1, \dots, K$, solve the equations 
\begin{align*}
L u_j + \epsilon^{-1} (1- \chi_{S_j^0}) u_j  & = d, \\
L^T v_j + \epsilon^{-1} (1- \chi_{S_j^0}) v_j  & = 1.
\end{align*}
\State  Normalize $u_j = \frac{u_j}{(1 + \epsilon \| u_j \|_{\ell^1})^2}$, $v_j = v_j$.
\State  Assign vertex $v$ to ${\vec S}^{t+1}_j$ where
\[
j=\mathrm{argmax} \{ u_1 \cdot v_1 (v) , \dots, u_K\cdot v_K (v) \}
\]
(that is, optimize $\nabla_\phi E$)
breaking ties randomly if needed.
\State Set $t = t+1$.
\EndWhile
\end{algorithmic}
\end{algorithm}

\subsubsection{Semi-supervised learning}

In cases where we have a labeled set of nodes $T$ with labels $\hat\phi_v \in \{ 0,1 \}$ indicating whether we want node $i$ to be in the subgraph ($\hat\phi_v = 1$) or its complement ($\hat\phi_v = 0$), we can incorporate this information into our approach as follows.

For the subgraph detector, we use
$E_{\epsilon,\lambda,T}(\phi) =  E_{\epsilon}(\phi) + \lambda \sum_{v \in T} \left( \phi_v - (1-\hat\phi_v) \right)^2$. Then the rearrangement algorithm needs to be modified at step 3 to become: Assign vertex $\ell$ to subgraph $S^1$ if $\nabla_\phi E$ is optimized
\[
\max_{|S| = k}  \ \frac1\epsilon \sum_{\ell \in S} u(\ell) \cdot v (\ell)   + 2 \lambda \sum_{v\in T}[\chi_S(v) - (1-\hat\phi_v)],
\]
where $\chi$ is the binary-valued indicator function. This again is solved by picking the largest elements
(we break ties by picking the lowest-index maximizers if needed).
Since the energy is still convex, the energy still increases at each iteration.

For the $K$-partitioner, we have a labeled set of nodes $T_i$ with labels $\hat\phi_{i,v} \in \{ 0,1 \}$, for $i = 1,\dots, K$ indicating whether we want node $v$ to be in partition element $i$, with $\sum_i \hat\phi_{i,v} = 1$ for $v\in \cup_i T_i$. We can incorporate this information into our approach by modifying the energy to be the concave functional
\begin{equation}
    \label{eqn:parEssl}
\tilde{E}_{\epsilon,\lambda}(\phi_1,\dots,\phi_K) = \tilde{E}_{\epsilon}(\phi_1,\dots,\phi_K) - \lambda \sum_{v \in T} \sum_{j=1}^K (\phi_{j,v} - (1 - \hat \phi_{j,v}))^2
\end{equation}
with the gradient rearrangement being appropriately modified.

\section{Numerical Results}
\label{s:NumRes}

We test the performance of these algorithms both on synthetic graphs and an assortment of ``real-world" graphs. 
For the synthetic tests, we use a particular set of undirected stochastic block models which we call the MultIsCale $K$-block Escape Ensemble (MICKEE), designed to illustrate some of the data features which our algorithms handle. 
A MICKEE graph consists of $N$ nodes partitioned into $K+1$ groups of sizes $N_1$, $\ldots$, $N_K$, and $N_{K+1} = N-\sum_{j=1}^K N_j$, where $N_1<N_2<\ldots<N_K<N_{K+1}$ (see the 2-MICKEE schematic in~\cref{fig:lopsided}). 
The nodes in the first $K$ groups induce densely connected Erd\H{o}s--R\'enyi (ER) subgraphs (from which we will study escape times) while the last group forms a sparsely connected ER background graph. 
Each of the $K$ dense subgraphs is sparsely connected to the larger background graph. 
The goal is to recover one of the planted subgraphs, generally the smallest. 
A na\"ive spectral approach will often find one of the planted graphs, but we know of no way to control which subgraph is recovered. Our subgraph detector method, in contrast, can be directed to look at the correct scale to recover a specific subgraph, as we will demonstrate in the 2-MICKEE example (i.e., with two planted subgraphs). 

\begin{figure}
    \centering
    \includegraphics[width=.35\textwidth]{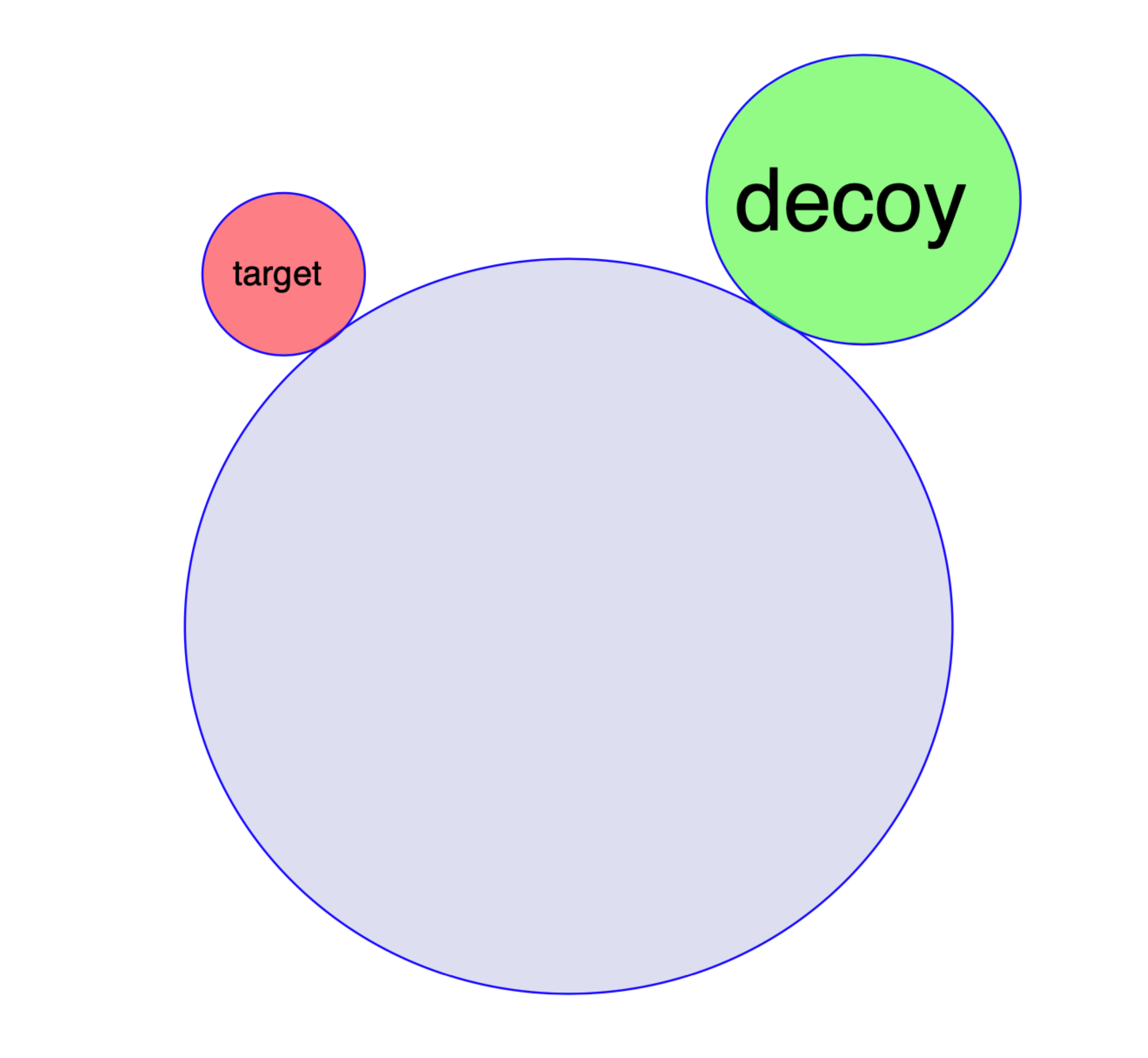}
    \caption{Schematic of a 2-MICKEE graph, with three dense subgraphs that are randomly connected to each other. Our subgraph detectors can identify the target subgraph, ignoring other planted subgraphs at different scales. Our partitioner correctly identifies each subgraph as a partition, regardless of the scale.}
    \label{fig:lopsided}
\end{figure}

We explore a number of variations on the basic MICKEE theme, including (1) making the large subgraph have a power law degree distribution (with edges drawn using a loopy, multi-edged configuration model), (2) adding more planted subgraphs with sizes ranging across several scales, (3) adding uniformly random noise edges across the entire graph or specifically between subgraphs, and (4) varying the edge weights of the various types of connections. For brevity, we refer to a MICKEE graph with $K$ planted subgraphs (not including the largest one) as a $K$-MICKEE graph.

\subsection{Subgraph Detection}
We explore the performance of~\Cref{alg:subgraph} using four benchmarks, which emphasize (1) noise tolerance, (2) multiscale detection, 
(3) robustness to heavy-tailed degree distributions, and 
(4) effective use of directed edges, respectively. In each of these tests, the target subgraph is the smallest planted subgraph.

\subsubsection*{Robustness to noise.} In~\cref{fig:3earnonlocal_sg} we visualize results from~\Cref{alg:subgraph} on $3$-MICKEE graphs, varying the amount and type of noise. While it is possible to get a bad initialization and thus find a bad local optimum the subgraph detector usually finds the target exactly, except in the noisiest regime (which occurs roughly at the point where the number of noise edges is equal to the number of signal edges).
 
\begin{figure}
\centering
\begin{subfigure}{.4\textwidth}
\includegraphics[width=\textwidth]{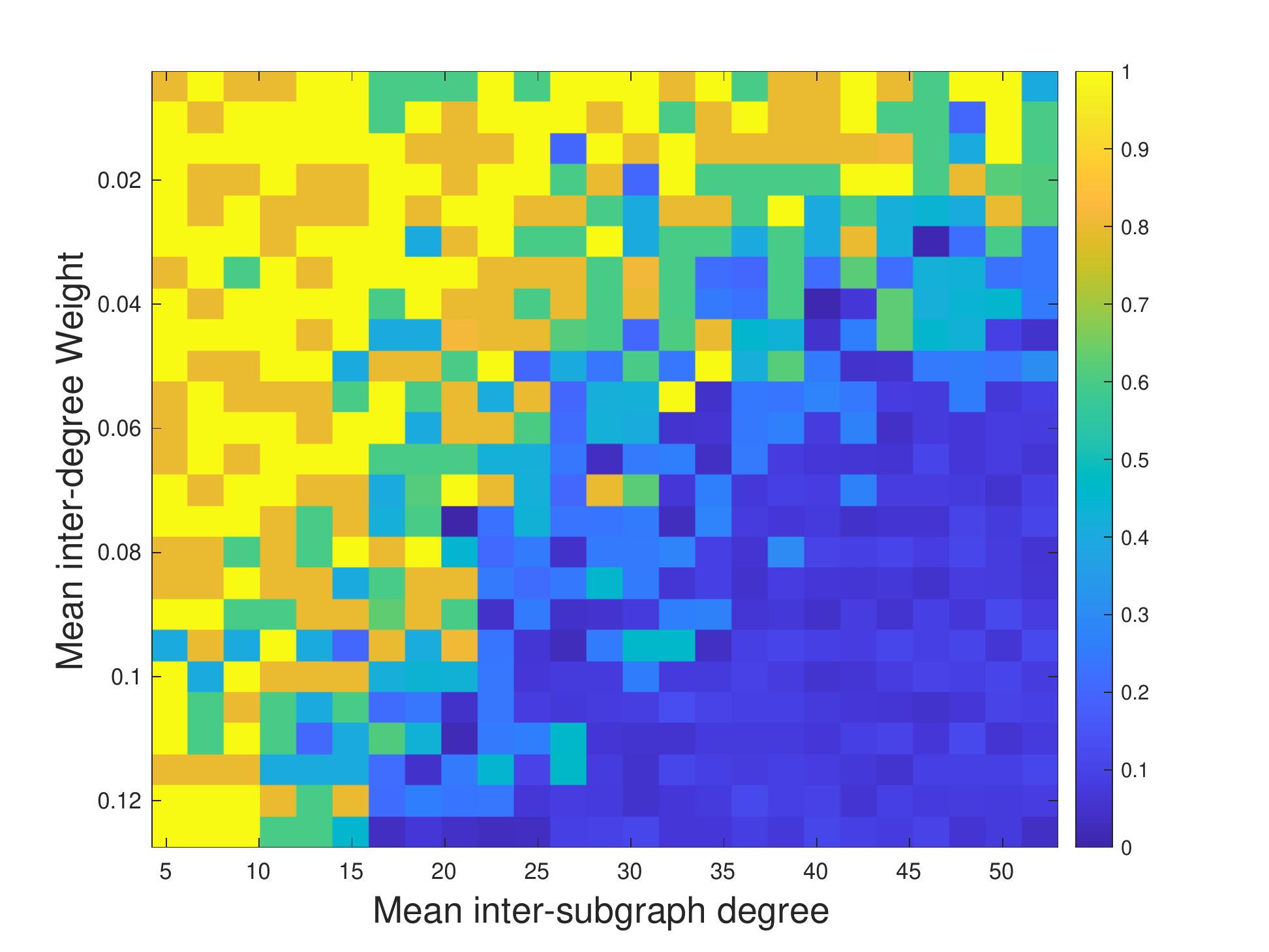}
\caption{Average of 5 runs}
\end{subfigure}
\begin{subfigure}{.4\textwidth}
\includegraphics[width=\textwidth]{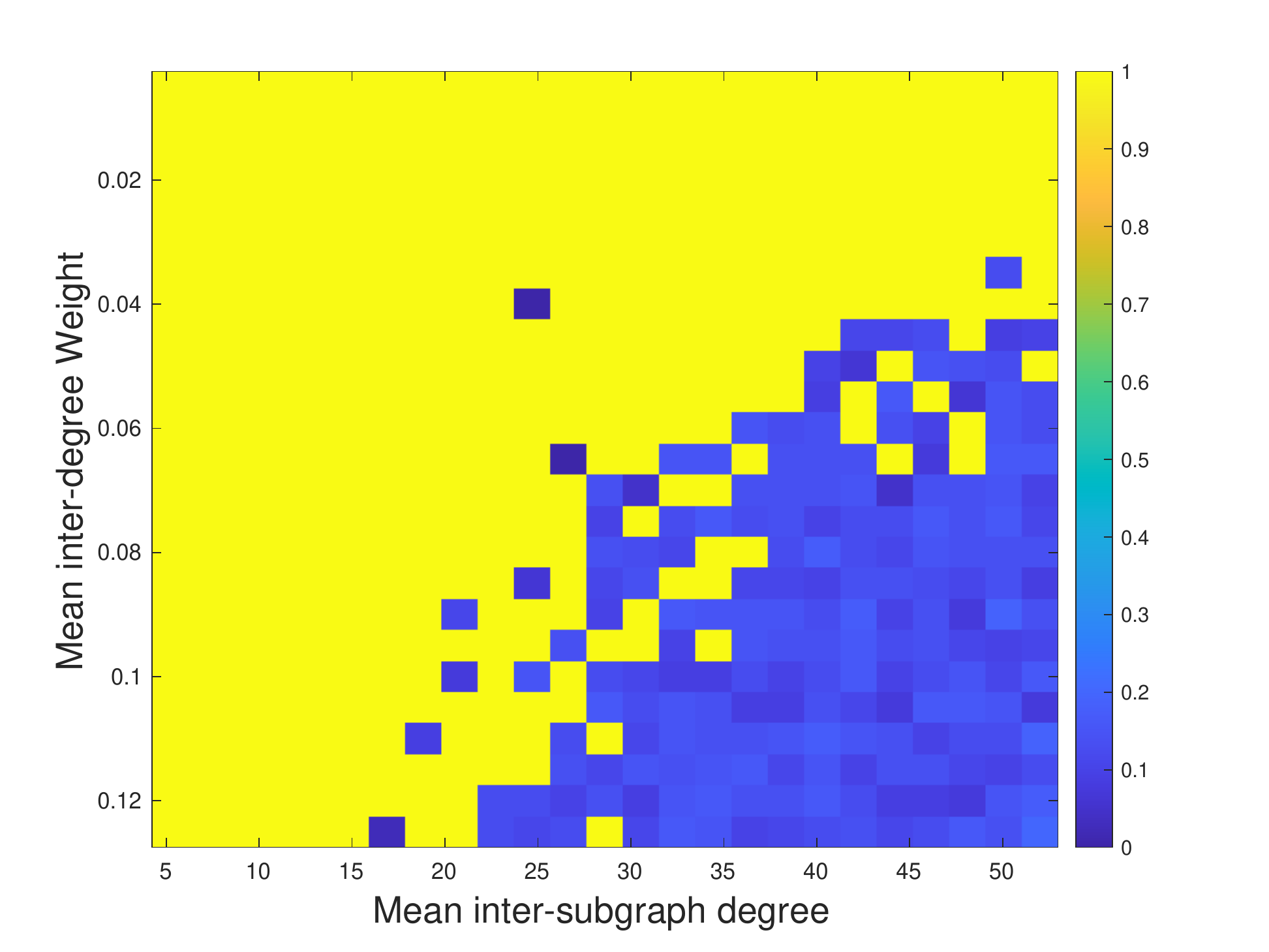}
\caption{Best of 5 runs}
\end{subfigure}
\caption{Accuracy of~\Cref{alg:subgraph} as a function of mean inter-subgraph degree (the mean taken over the nodes of the target subgraph) and mean weight of the inter-component edges (not including non-edges) for $3$-MICKEE graphs with planted subgraphs of sizes $80$, $160$, and $240$ nodes, with a total of $1,000$ nodes in the entire graph. The expected in-subgraph-degree is fixed at $20.8$ (with intra-component edge weights given by $1$). Inter-group edge weights are drawn from a normal distribution with maximum ranging from $.01$-$.25$. As long as the noise level is not too high, the subgraph detector finds the smallest planted subgraph despite the presence of ``decoy'' subgraphs at larger scales. This may be contrasted with spectral clustering, which is attracted to the larger scales.}
\label{fig:3earnonlocal_sg}
\end{figure}

\subsubsection*{Range of scales.} 
We generated $2$-MICKEE graphs with varying sizes of the subgraphs relative to each other and the total mass.  We take $1500<N<2500$ for the total size and vary the percentage of smallest planted subgraph as $.02N\leq N_1 \leq .15N$ with $N_2 = 2 N_1$.  Here, the inter-edge density was set to $.01$ (in-subgraph-degree values between $(1-3p)*N*.01$ for $.02<p<.15$) with mean inter-edge weight $.05$ compared to intra-group edge weights of $1$.  We used this framework to assess the detectability limits of sizes of the smallest components, and numerically we observe that small communities are quite detectable using our algorithm.  Using the best result over $5$ initializations, we were able to detect the smallest ear over the entire range and we did so reliably on average as well. Since the resulting figure would thus not be terribly informative for this range, we forego including a similar heat plot over this range of parameters.   


\subsubsection*{Heavy-tailed degree distributions.} For the results in \cref{fig:3earpowerlaw_sg}, we use a power law degree distribution in the largest component of $3$-MICKEE graphs with $N_1 = 80, N_2 = 160, N_3 = 240$ and $N = 1000$. Surprisingly (at least to us), smaller power-law exponents (corresponding to more skewed degree distributions) actually make the problem much easier (whereas adding noise edges had little effect). We conjecture that this is because, in the presence of very high-degree nodes, it is difficult to have a randomly occurring subgraph with high mean escape time, since connections into and out of the hubs are difficult to avoid.

\begin{figure}
\centering
\begin{subfigure}{.4\textwidth}
\includegraphics[width=\textwidth]{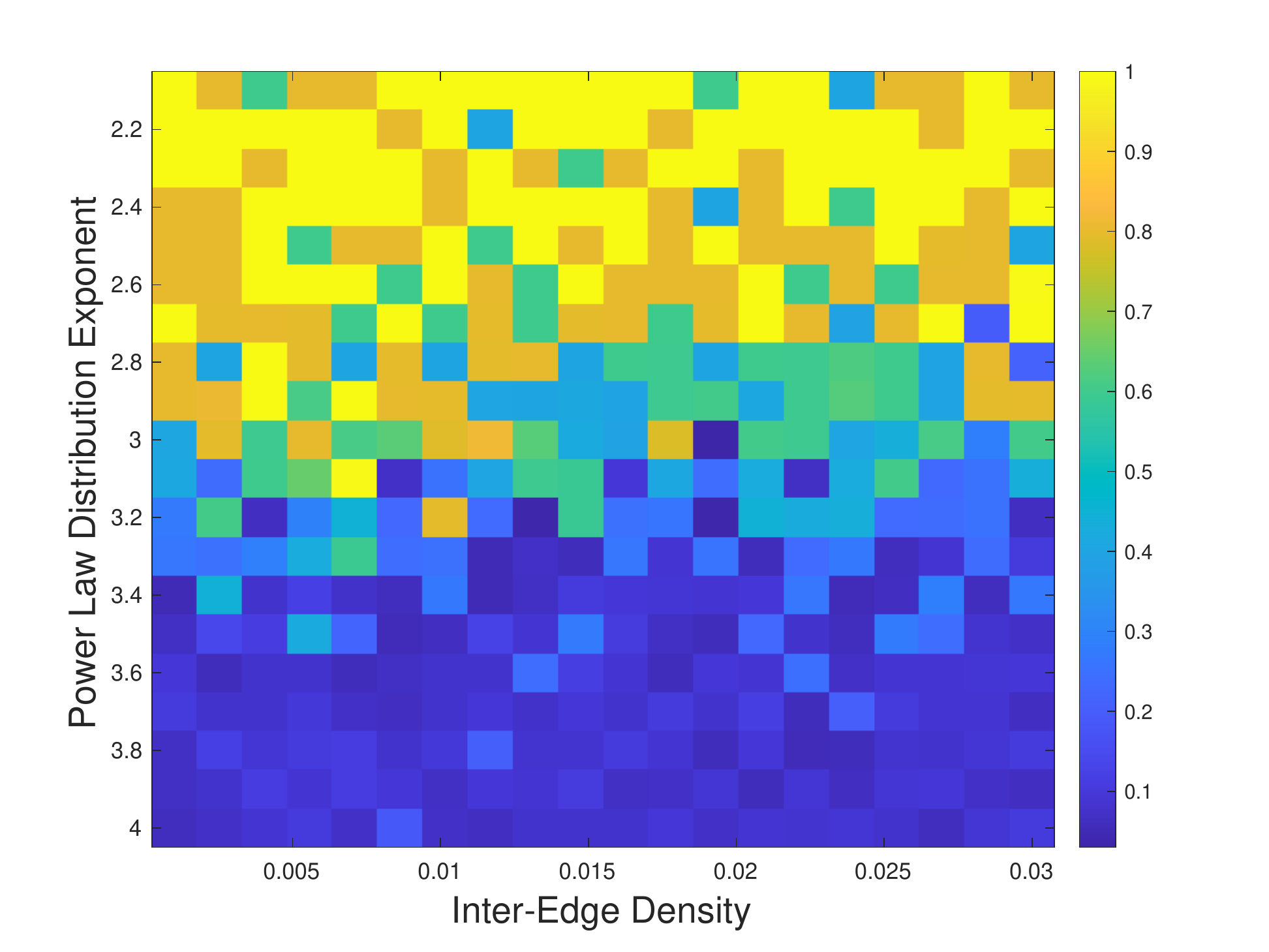}
\caption{Average of 5 runs}
\end{subfigure}
\begin{subfigure}{.4\textwidth}
\includegraphics[width=\textwidth]{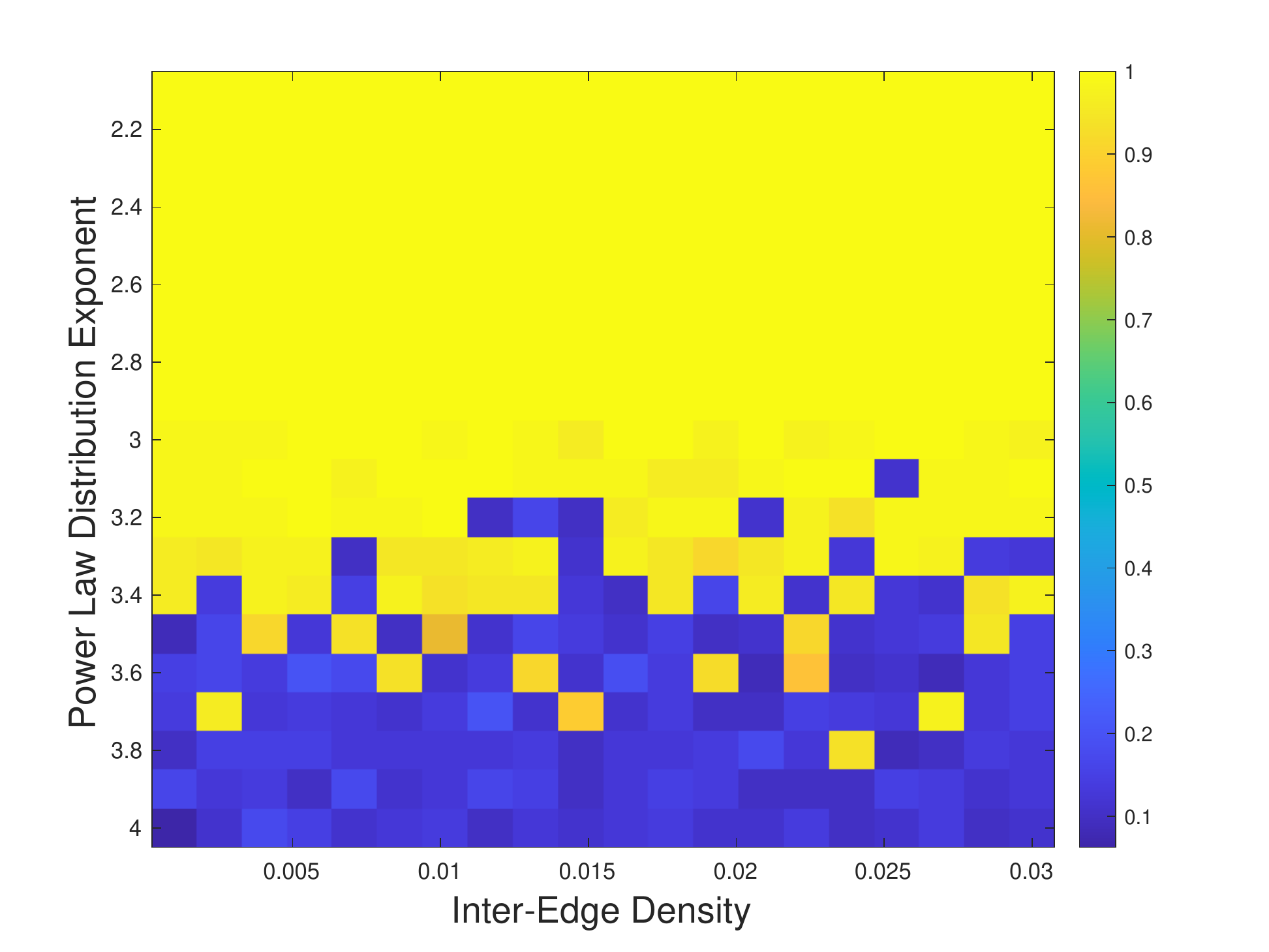}
\caption{Best of 5 runs}
\end{subfigure}
\caption{Accuracy of~\Cref{alg:subgraph} on a $3$-MICKEE graph with a power law distribution as a function of the power law exponent and inter-cluster edge density.  We observe a robustness to both the exponent and density (especially in the right panel) up to a sharp cutoff around 3.4. Note the low exponents (typically considered to be the harder cases) are actually easier in this problem.}
\label{fig:3earpowerlaw_sg}
\end{figure}

\subsubsection*{Directed edge utilization.} In~\cref{fig:ERcycle_sg} we consider the problem of detecting a directed cycle appended to an ER graph. The graph weights have been arranged so that the expected degree of all nodes is roughly equal. There are many edges leading from the ER graph into the cycle, with only one edge leading back into the ER graph. This makes the directed cycle a very salient dynamical feature, but not readily detectable by undirected (e.g.\ spectral) methods. We considered a large number of cycle sizes relative to the ER graph and with a proper choice of $\epsilon$, we were able to detect the cycle in all cases. Thus, this detector finds directed components very robustly due to the nature of the escape time.

\begin{figure}
\centering
\includegraphics[width=0.25\textwidth]{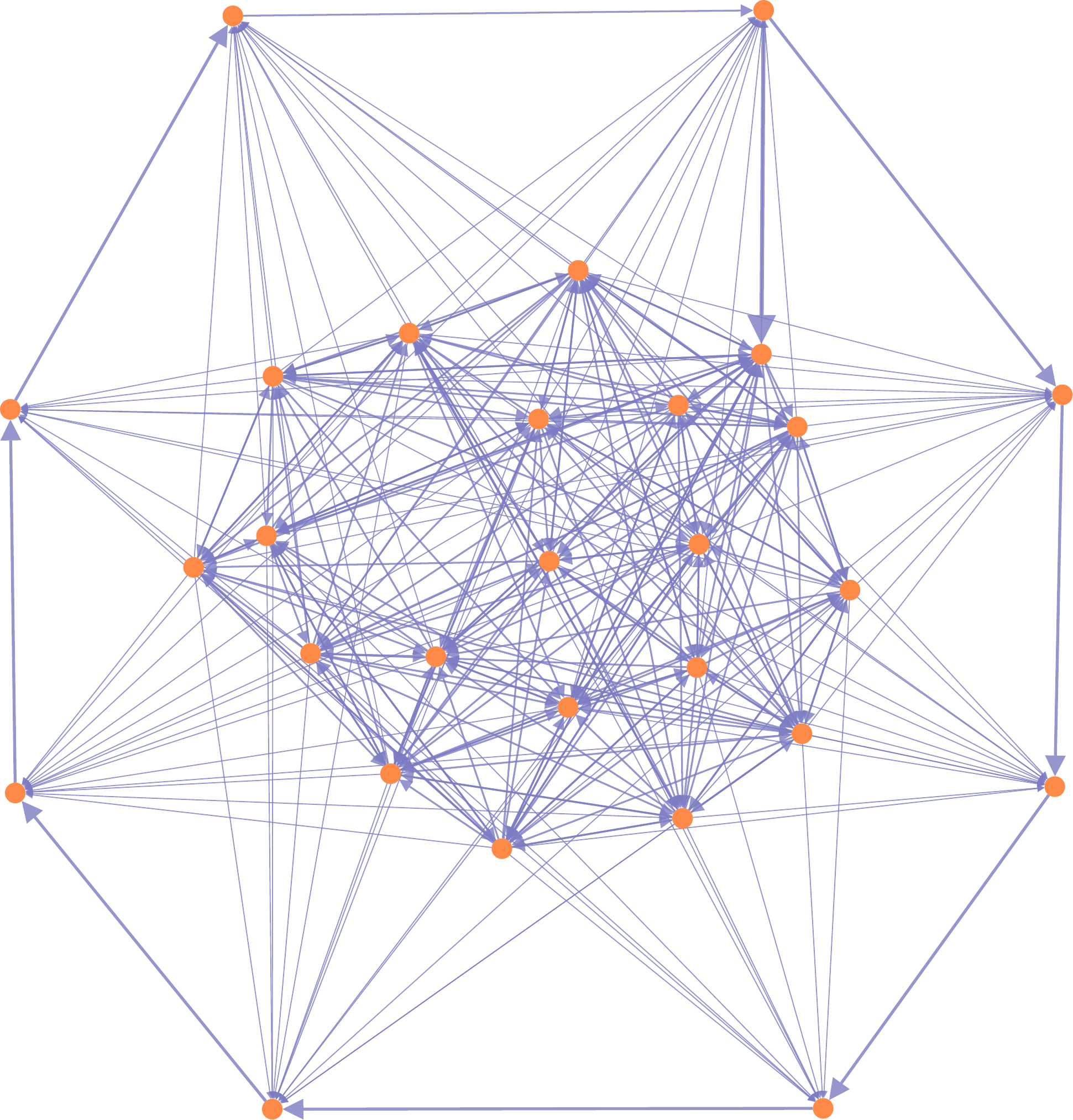}
\caption{A directed ER graph with a directed cycle appended. Note that there is only one edge (in the upper left) leading from the cycle to the ER graph, with many edges going the other direction from the ER graph to the cycle. The cycle nodes have the same expected degree as the ER nodes, yet a random walker would naturally get stuck in the cycle for a long time.  Detecting such a dynamical trap is a challenge for undirected algorithms, but~\Cref{alg:subgraph} detects it consistently over a wide range of cycle lengths and ER graph sizes.}
\label{fig:ERcycle_sg}
\end{figure}

\subsubsection*{Variation over choice of $N_1$.} In Figure \ref{fig:EK}, we consider how the Mean Exit Time as well as the regularized energy in \eqref{e:l1energy} behaves as we vary the constrained volume of our algorithm.  We considered a $2$-MICKEE graph with $N_1 = 50$, $N_2 = 100$ and $N = 1000$.  We took the baseline ER density $.03$ and the inter-edge density was set to $.025$ with mean inter-edge weight $.1$.

\begin{figure}
\centering
\begin{subfigure}{.4\textwidth}
\includegraphics[width=\textwidth]{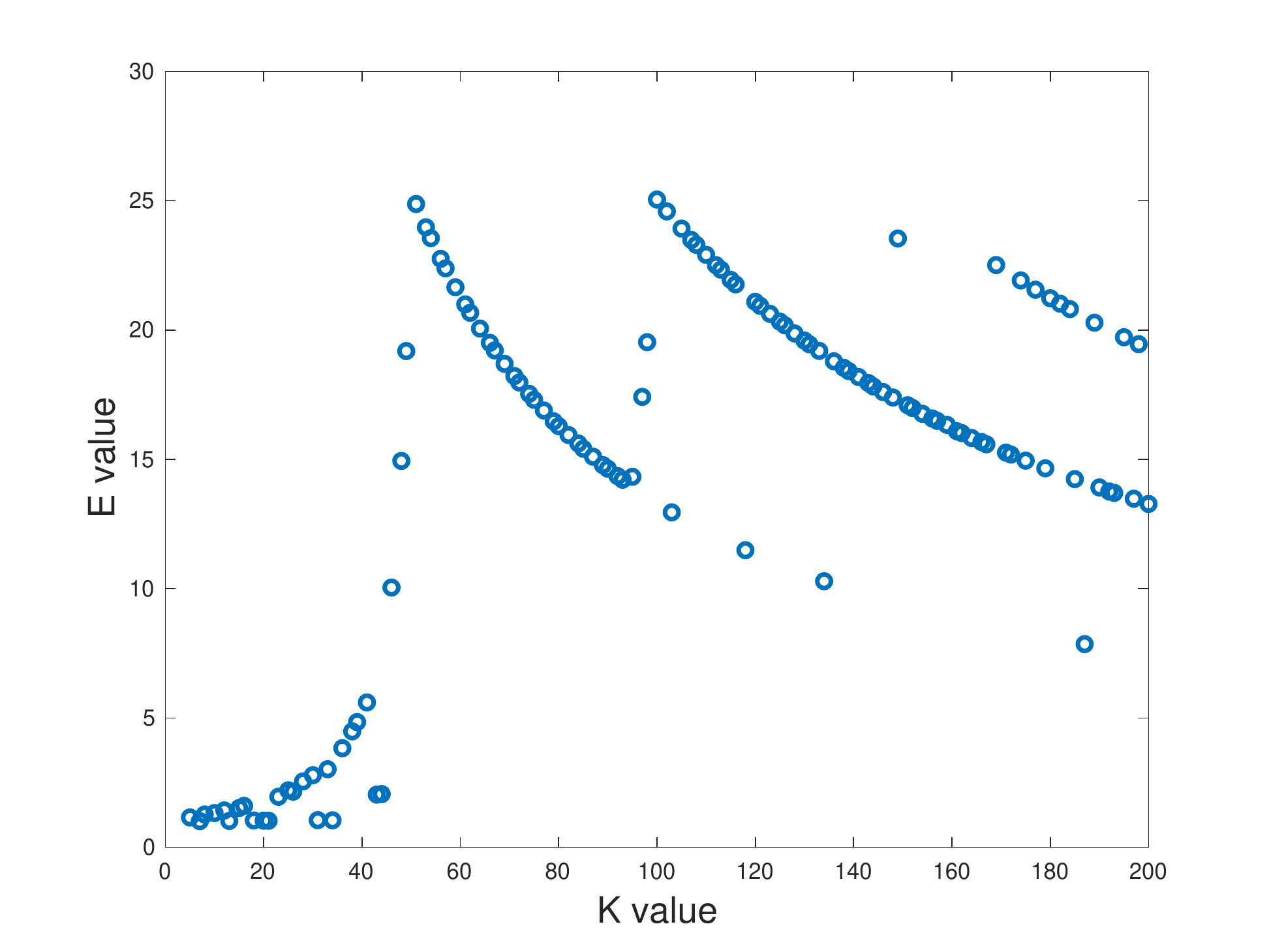}
\caption{True mean exit time}
\end{subfigure}
\begin{subfigure}{.4\textwidth}
\includegraphics[width=\textwidth]{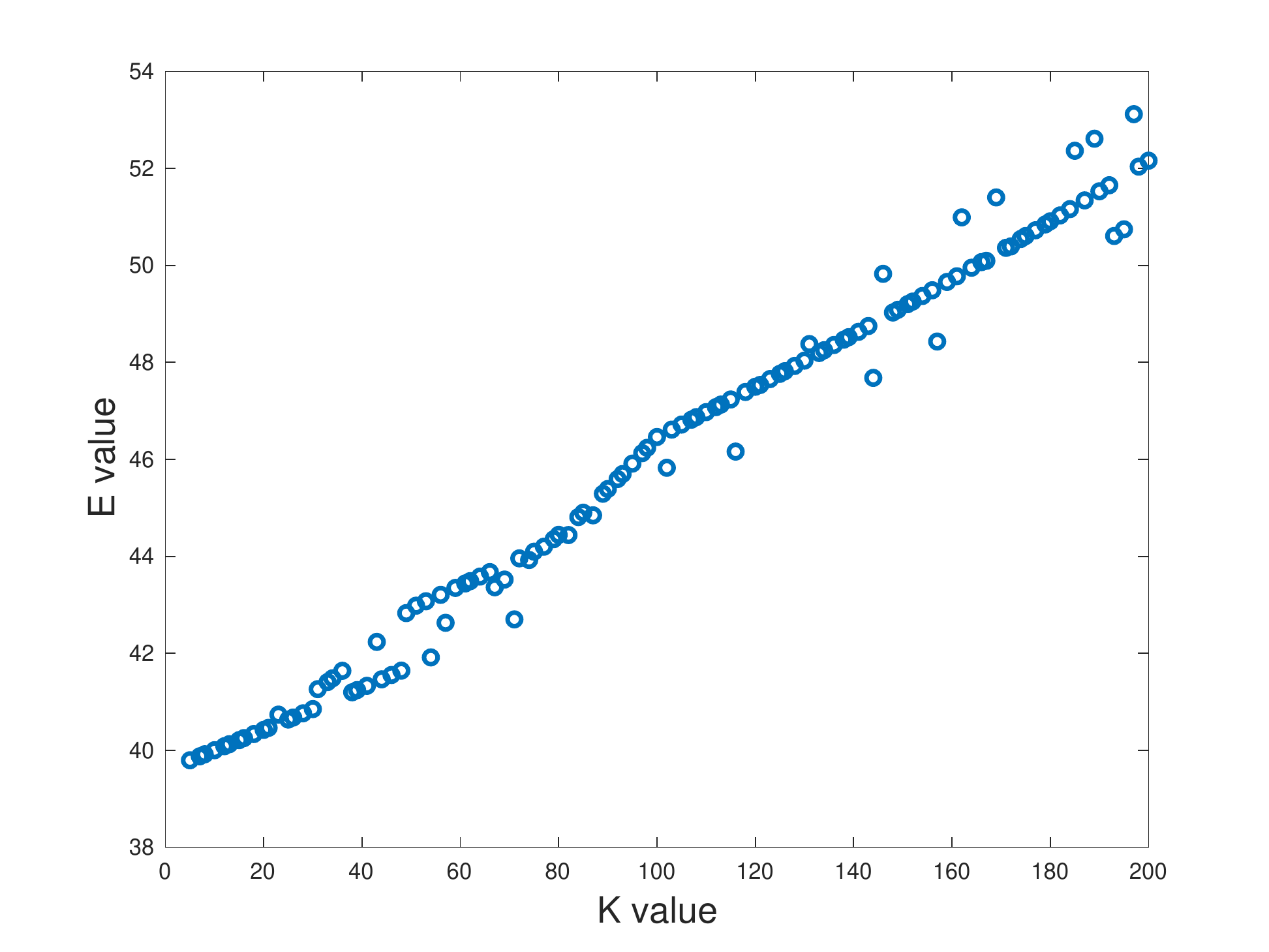}
\caption{Regularized energy}
\end{subfigure}
\caption{The score of the optimal sub-graph found with~\cref{alg:subgraph}.  Both plots have clear shifts near $k = 50$ corresponding to the smallest component and $k = 100$ corresponding to the second smallest component. This suggests that the size of natural subgraphs within a given graph can be detected from breaks in the subgraph scores as the size of the target in~\cref{alg:subgraph} varies.}
\label{fig:EK}
\end{figure}

In summary, we find that the subgraph detector is able to robustly recover planted communities in synthetic graphs and is robust to a range of application-relevant factors.

\subsection{\texorpdfstring{$K$}{K}-partition method}

We will now consider the performance of \cref{alg:partitioner} in a variety of settings.  Throughout, we will give heat plots over the variation of the parameters to visualize the purity measure of our detected communities from our ground-truth smallest component of the graph, over $5$ iterations of the algorithm.  The purity measure is 
\[
\frac{1}{N} \sum_{k=1}^K \max_{1 \leq l \leq K} N_k^l
\]
for $N_k^l$ the number of data samples in cluster $k$ that are in ground truth class $l$.  

In \cref{fig:3earnonlocal} we consider a $\rho-\Delta$ heat plot of the purity measure for a 4-partition of a $3$-MICKEE graph using delocalized connections with $N_1 = 80, N_2 = 160, N_3 = 240$ and $N = 1000$, varying the density of the inter-community edge connections ($0 < \rho < .1$) and the mean weight of the inter-component edges ($0<\Delta<.125$). We vary over number and strength of connecting edges between components and consider the purity measure as output.

\begin{figure}
\centering
\begin{subfigure}{.4\textwidth}
\includegraphics[width=\textwidth]{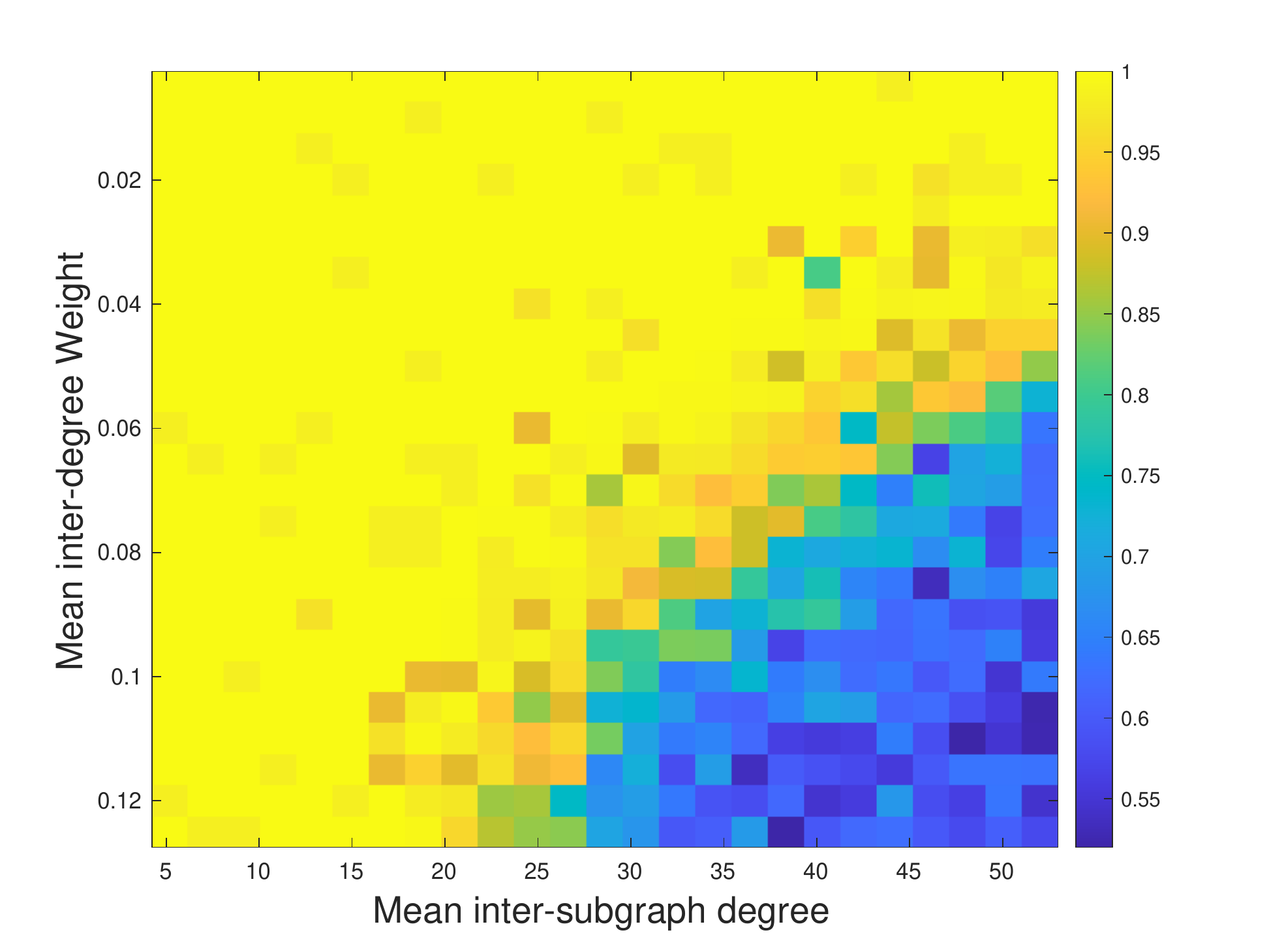}
\caption{Average of 5 runs}
\end{subfigure}
\begin{subfigure}{.4\textwidth}
\includegraphics[width=\textwidth]{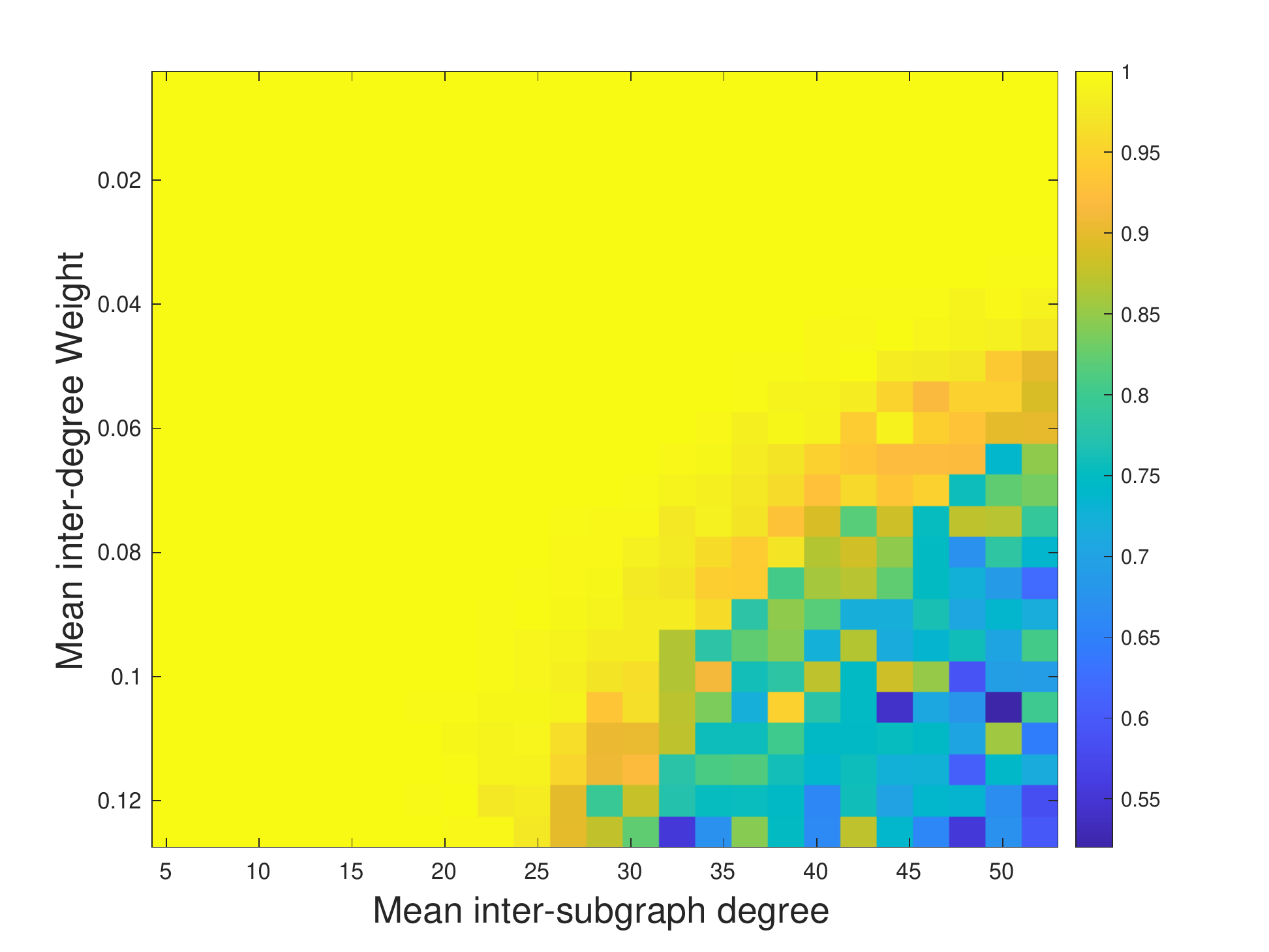}
\caption{Best of 5 runs}
\end{subfigure}
\caption{The purity measure for~\cref{alg:partitioner} on $3$-MICKEE graphs.  We vary the density of the inter-region edges and their edge weights. We observe robust (usually perfect) detection over a range of these parameters, with a sharp cutoff (especially in the left panel) when the noise levels grow too high, suggesting that detection is still possible beyond this cutoff, but the energy landscape has more bad local optima beyond this point.}
\label{fig:3earnonlocal}
\end{figure}

In addition, we have tested \cref{alg:partitioner} on MICKEE graphs with varying sizes of the components relative to each other and the total mass where the connections between ER graphs include more random edges with weak connection weights.  \Cref{fig:2earnonlocal} shows results from testing the algorithm on $2$-MICKEE graphs with varying sizes of the components relative to each other and the total mass.  We take $1500<N<2500$ for the total size and vary the percentage of smallest planted subgraph as $.02N\leq N_1 \leq .15N$ with $N_2 = 2 N_1$.  Here, the inter-edge density was set to $.025$ with mean inter-edge weight $.05$. The question addressed in this experiment is how small can we get the components and still detect them. We heat map the average purity measure varying the number of vertices in the graph and the relative size of the smallest sub-graph (i.e., $N_1/N$). 

\begin{figure}
\centering
\begin{subfigure}{.4\textwidth}
\includegraphics[width=\textwidth]{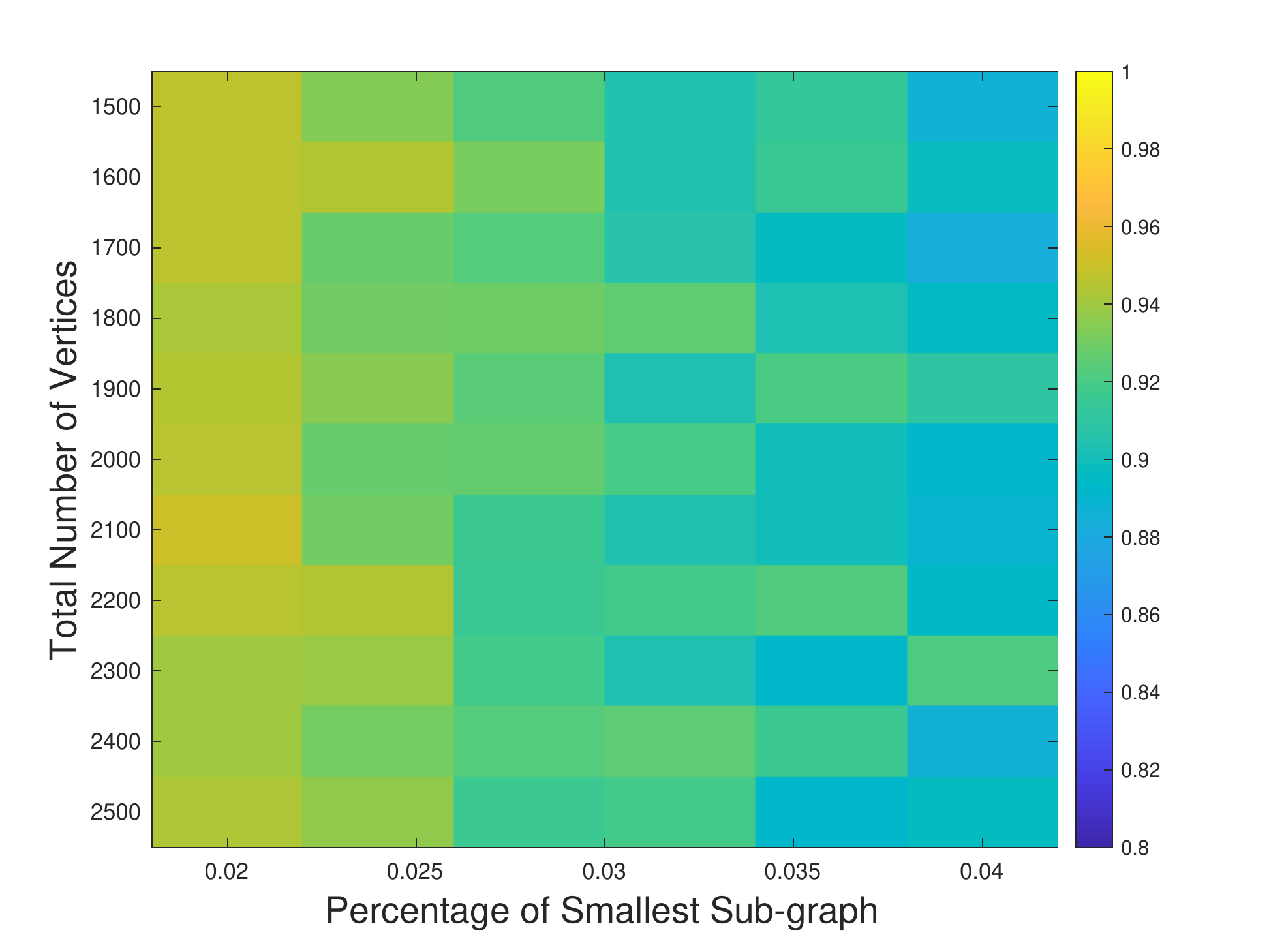}
\caption{Average of 5 runs}
\end{subfigure}
\begin{subfigure}{.4\textwidth}
\includegraphics[width=\textwidth]{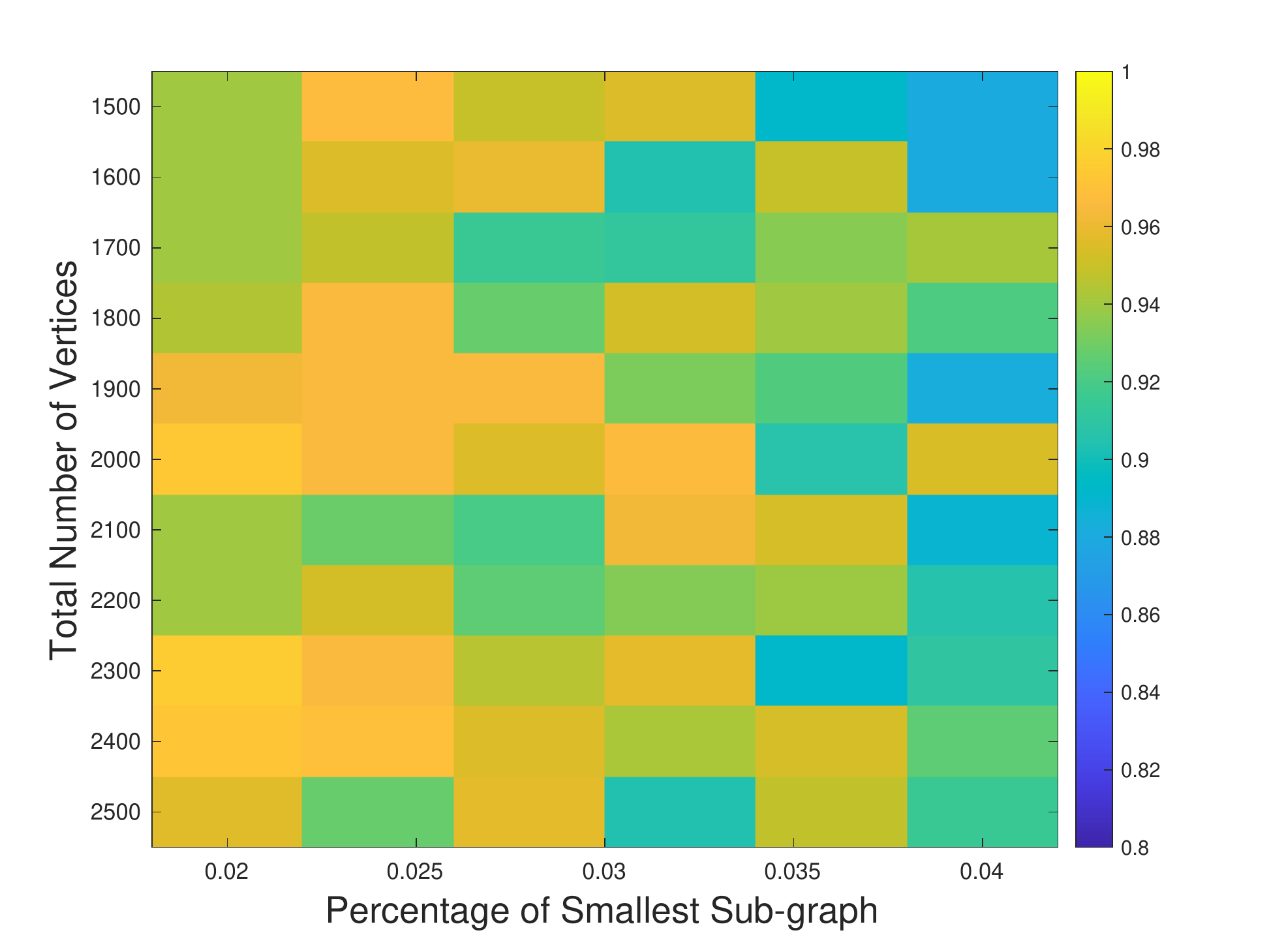}
\caption{Best of 5 runs}
\end{subfigure}
\caption{The purity measure for the partitioner acting on a $2$-MICKEE graph with the fraction of nodes in the smaller planted subgraph varying, along with the size of the graph. We observe a generally robust partitioning.}
\label{fig:2earnonlocal}
\end{figure}

We similarly consider the partitioning problem on a version of the $3$-MICKEE graph with power-law degree distribution in the largest component, using delocalized connections with $N_1 = 80, N_2 = 160, N_3 = 240$ and $N = 1000$. \Cref{fig:3earpowerlaw} provides a $\rho-q$ plot for results from varying the density ($.001<\rho<.03$) of the edge-density of connections between the components of the graph, using a power law degree distribution for the largest component with exponent ($2.1 \leq q \leq 4$). 

\begin{figure}
\centering
\begin{subfigure}{.4\textwidth}
\includegraphics[width=\textwidth]{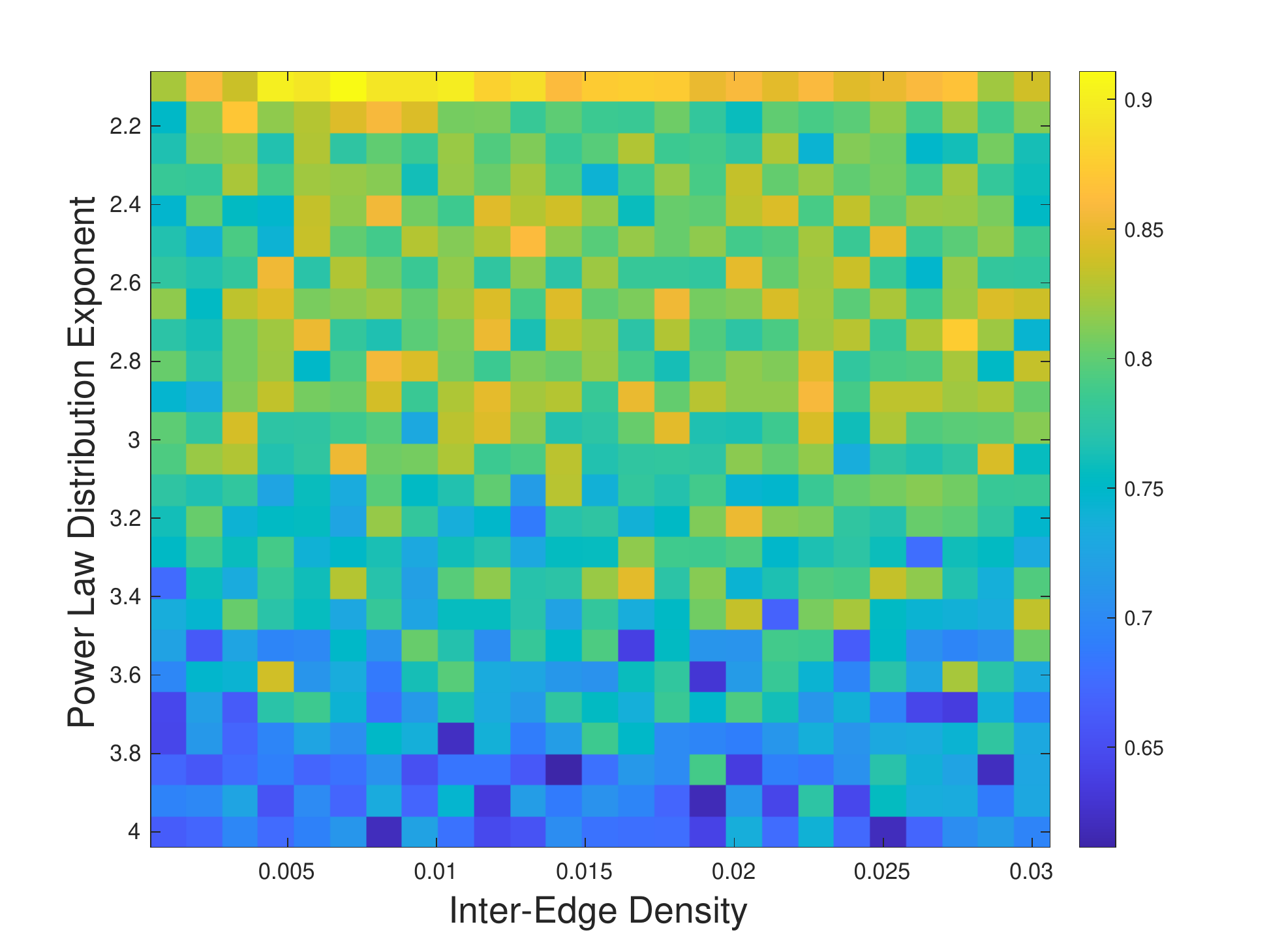}
\caption{Average of 5 runs}
\end{subfigure}
\begin{subfigure}{.4\textwidth}
\includegraphics[width=\textwidth]{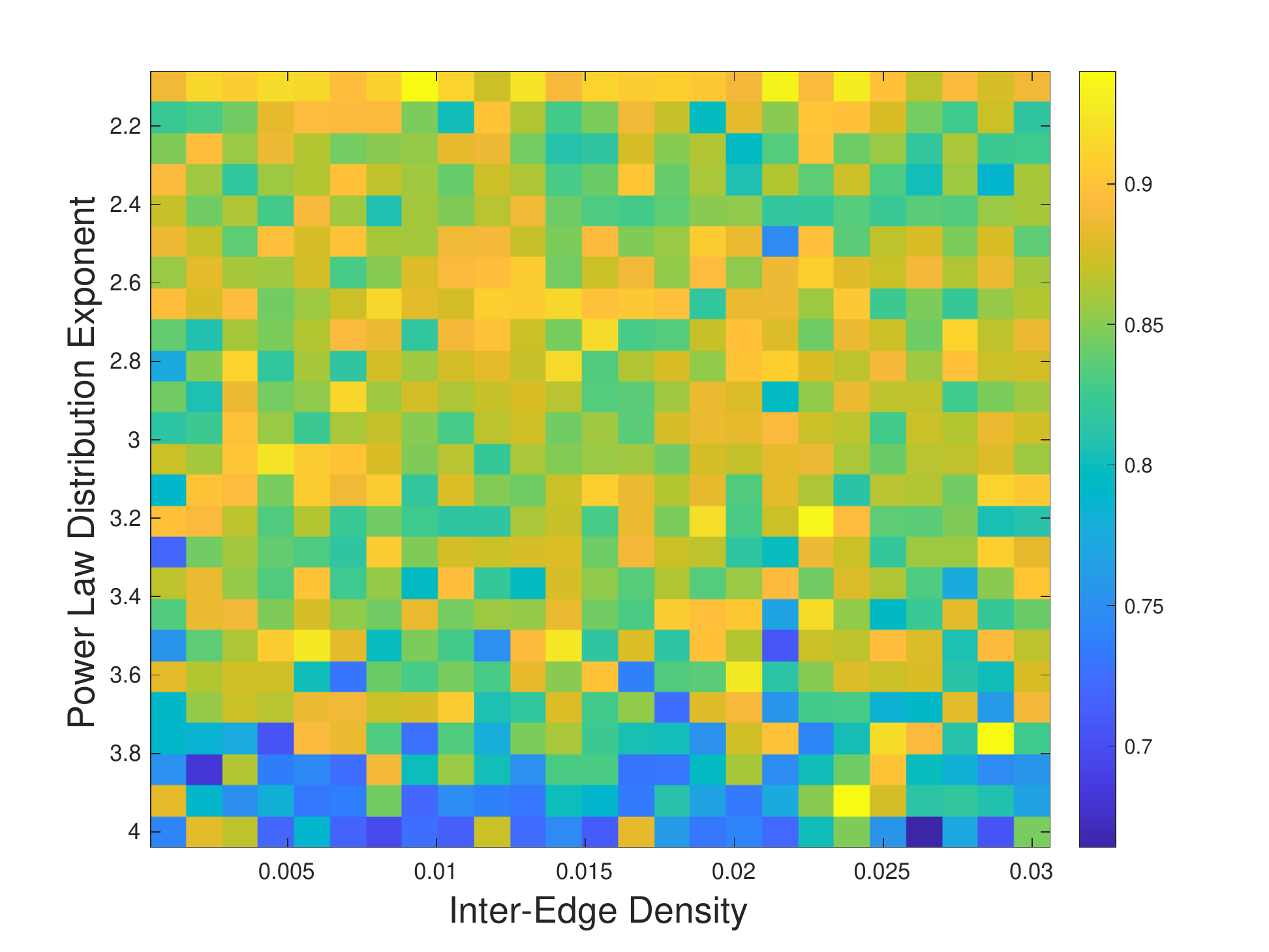}
\caption{Best of 5 runs}
\end{subfigure}
\caption{Purity achieved by~\cref{alg:partitioner} on $3$-MICKEE graphs with a power law degree distribution, varying the exponent of the power law and inter-subgraph edge density. We observe generally robust partitioning (especially in the right panel).}
\label{fig:3earpowerlaw}
\end{figure}

\subsubsection{Graph clustering examples} 

We consider the family of examples as in \cite{yang2012clustering} and compare the best presented purity measures from that paper to a number of settings using our algorithms.  Since some of these examples are by their nature actually directed data sets, we throughout computed both the  directed and undirected adjacency matrix representations as appropriate to test against.  We ran the $K$-partitioner over a variety of scenarios for both cases.  In all these runs, we chose the value of $K$ to agree with the metadata (we avoid term "
ground truth", as the node labels themselves may be noisy or not the only good interpretation of the data). However, we note that our algorithm also does a good job in a variety of settings selecting the number of partitions to fill without precisely providing this correct number \emph{a priori}.

For our study, we consider a number of various options for the algorithm.  First, the initial seeding sets were chosen either uniformly at random or using $K$-means on the first $K$-eigenvectors of the graph Laplacian.  We consider the best result over $10$ outcomes.  In addition, we considered a range of values of $\epsilon$, all of which were a multiplicative factor of the inverse of the Frobenius norm of the graph Laplacian, denoted $\| L \|_{{\rm Fro}}$, which sets a natural scaling for separation in the underlying graph.  See for instance the related choice in \cite{osting2014minimal}.  We computed a family of partitions for $\epsilon = {50 \nu}/{\| L \|_{{\rm Fro}}}$, where $\nu = e^{.2 \ell}$ with $-50 < \ell < 50$.  Finally, we also considered the impact of semi-supervised learning by toggling between $\lambda = 0$ and $\lambda = 10^6$ in \cref{eqn:parEssl} with $10$\% of the nodes being included in the learning set.  Clearly, there are many ways we might improve the outcomes, by for instance increasing the number and method of initialization and  refining our choices of $\epsilon$ or $\lambda$; nevertheless, we see under our current choices that our fast algorithm performs well over a range of such parameters, as reported in Table \ref{tab:purmeas}.

For each data set in Table \ref{tab:purmeas}, we report the best outcome using directed adjacency matrices to build the Graph Laplacian using both the $K$-means and random initializations but with no semi-supervised learning (Directed); the best outcome using symmetrized adjacency matrices to build the Graph Laplacian using both the $K$-means and random initializations but with no semi-supervised learning (Undirected); the best outcome when Semi-Supervised Learning is turned on over any configuration (Semi-supervision), the $K$-means only outcome ($K$-means only) and the best data from all the experiments reported in \cite{yang2012clustering} (Best from \cite{yang2012clustering}). Our results promisingly demonstrate that our fast algorithm is very successful in many cases in discovering large amounts of community structure that agrees with the metadata in these explicit data sets.  Given that our communities are all built around random walks in the graph, it is not clear that all ground-truth designated communities would align well with our methods. For example, we note that our results do not align well with the metadata in the {\rm POLBLOGS} data set.  A major takeaway from the table, however, is that in several examples we see that using the directed nature of the data provides better agreement with the metadata (as indicated by the green cells).  Perhaps most striking in the table is that the best run of our fast algorithm, even without semi-supervised learning, provides better agreement with the metadata than \cite{yang2012clustering} for many of the data sets.

As a statistical summary of our findings, we had in total 
$39$ directed datasets and $6$ undirected data sets that came from a variety of domains (image, social, biological, physical, etc.).
The networks are sized between $35$ nodes and $98,528$ nodes, having $2$-$65$ classes per network.  Among directed networks, 
$21$ data sets gave highest purity with the metadata with semi-supervised learning turned on, while $13$ have the best result from \cite{yang2012clustering}, and $3$ have $K$-means only best.  For $9$ total data sets (green in the table), the directed version of our algorithm is better than the symmetrized undirected version, while $5$ are tied (yellow) and for $25$ the undirected method is better (orange).  When \cite{yang2012clustering} is best, the median gap from our result with semi-supervised learning is $.05$.
When our algorithm with semi-supervised learning is best, the median gap from \cite{yang2012clustering} is $.05$.  There is no clear relationship between data domain and performance or node count and performance.  However, semi-supervision generally did improve the results the most with a smaller class count (median $3$) versus \cite{yang2012clustering} (median $20$).

When the directed algorithm is better than the undirected version, the median gap is $0.03$. Interestingly, $5$ of the datasets where directed was better were image or sensor data, with the two largest gaps ($.07$ and $.11$) being digit datasets.
When undirected was better, the median gap was $0.06$, with the largest gap being $.29$, for the 20NEWS dataset.  When semi-supervision improves over our method (max of directed and undirected performance), the median improvement is $.06$, and the max improvements were $.22$ and $.20$.  There is no obvious relationship between edge density and algorithm performance.

\begin{footnotesize}

\begin{table}
\centering

\sisetup{detect-weight,mode=text}
\renewrobustcmd{\bfseries}{\fontseries{b}\selectfont}

\begin{tabular}{llr>{\raggedleft}p{.3in}p{.3in}p{.3in}p{.3in}p{.3in}p{.3in}p{.3in}}
\rot{Network} & \rot{Domain} & \rot{Vertices} & \rot{Density} & \rot{Classes} & \rot{Directed} & \rot{Undirected} & \rot{Semi-supervision}  & \rot{$K$-means only} & \rot{Best from \cite{yang2012clustering}} \\ 
\midrule
\multicolumn{8}{l}{\bf Directed data}\\
MNIST       & Digit     & 70,000    & 0.00 & 10& \colorbox{green}{0.85}  & \colorbox{green}{0.78}          & \textbf{0.98}                       & 0.84          & 0.97\\ 
VOWEL       & Audio     & 990       & 0.01 & 11& \colorbox{green}{0.35}  & \colorbox{green}{0.32}          & \textbf{0.44}                       & 0.34          & 0.37 \\ 
FAULTS      & Materials & 1,941     & 0.00 & 7 & \colorbox{green}{0.44}  & \colorbox{green}{0.42}          & \textbf{0.49}                       & 0.39         & 0.41 \\ 
SEISMIC     & Sensor    & 98,528    & 0.00 & 3 & \colorbox{green}{0.60}  & \colorbox{green}{0.59}          & \textbf{0.66}                       & 0.58          & 0.59 \\ 
7Sectors    & Text      & 4,556     & 0.00 & 7 & \colorbox{green}{0.27}  & \colorbox{green}{0.26}          & \textbf{0.39}                       & 0.26          & 0.34 \\ 
PROTEIN     & Protein   & 17,766    & 0.00 & 3 & \colorbox{green}{0.47}  & \colorbox{green}{0.46}          & \textbf{0.51}                       & 0.46          & 0.50 \\ 
KHAN        & Gene      & 83        & 0.06 & 4 & \colorbox{yellow}{0.59} & \colorbox{yellow}{0.59}          & \textbf{0.61}                       & 0.59          & 0.60 \\ 
ROSETTA     & Gene      & 300       & 0.02 & 5 & \colorbox{yellow}{0.78} & \colorbox{yellow}{0.78}          & \textbf{0.81}                       & 0.77          & 0.77 \\ 
WDBC        & Medical   & 683       & 0.01 & 2 & \colorbox{yellow}{0.65} & \colorbox{yellow}{0.65}          & \textbf{0.70}                       & 0.65          & 0.65 \\ 
POLBLOGS    & Social    & 1,224     & 0.01 & 2 & \colorbox{yellow}{0.55} & \colorbox{yellow}{0.55}          & \textbf{0.59}                       & 0.51          & NA \\
CITESEER    & Citation  & 3,312     & 0.00 & 6 & \colorbox{orange}{0.28} & \colorbox{orange}{0.29}          & \textbf{0.49}                       & 0.25          & 0.44\\  
SPECT       & Astronomy & 267       & 0.02 & 3 & \colorbox{orange}{0.79} & \colorbox{orange}{0.80}          & \textbf{0.84}                       & 0.79          & 0.79 \\ 
DIABETES    & Medical   & 768       & 0.01 & 2 & \colorbox{orange}{0.65} & \colorbox{orange}{0.67}          & \textbf{0.74}                       & 0.65          & 0.65 \\ 
DUKE        & Medical   & 44        & 0.11 & 2 & \colorbox{orange}{0.64} & \colorbox{orange}{0.68}          & \textbf{0.73}                       & 0.52          & 0.70 \\ 
IRIS        & Biology   & 150       & 0.03 & 3 & \colorbox{orange}{0.87} & \colorbox{orange}{0.90}          & \textbf{0.97}                       & 0.67          & 0.93 \\ 
RCV1        & Text      & 9,625     & 0.00 & 4 & \colorbox{orange}{0.35} & \colorbox{orange}{0.40}          & \textbf{0.62}                       & 0.32          & 0.54 \\ 
CORA        & Citation  & 2,708     & 0.00 & 7 & \colorbox{orange}{0.33} & \colorbox{orange}{0.39}          & \textbf{0.50}                       & 0.32          & 0.47 \\ 
CURETGREY   & Image     & 5,612     & 0.00 & 61& \colorbox{orange}{0.23} & \colorbox{orange}{0.29}          & \textbf{0.33}                       & 0.22          & 0.28\\ 
SPAM        & Email     & 4,601     & 0.00 & 2 & \colorbox{orange}{0.64} & \colorbox{orange}{0.70}          & \textbf{0.73}                       & 0.61          & 0.69 \\ 
GISETTE     & Digit     & 7,000     & 0.00 & 2 & \colorbox{orange}{0.87} & \colorbox{orange}{0.94}          & \textbf{0.97}                       & 0.81          & 0.94 \\ 
WEBKB4      & Text      & 4,196     & 0.00 & 4 & \colorbox{orange}{0.42} & \colorbox{orange}{0.53}          & \textbf{0.66}                       & 0.40          & 0.63 \\ 
CANCER      & Medical   & 198       & 0.03 & 14& \colorbox{orange}{0.49} & \colorbox{orange}{\textbf{0.55}} & 0.54                                & 0.45          & 0.54 \\ 
YALEB       & Image     & 1,292     & 0.00 & 38& \colorbox{orange}{0.44} & \colorbox{orange}{\textbf{0.54}} & 0.52                                & 0.41          & 0.51 \\ 
COIL-20     & Image     & 1,440     & 0.00 & 20& \colorbox{orange}{0.74} & \colorbox{orange}{\textbf{0.85}}          & 0.78                                & 0.82 & 0.81 \\ 
ECOLI       & Protein   & 327       & 0.02 & 5 & \colorbox{orange}{0.79} & \colorbox{orange}{\textbf{0.83}} & 0.81                                & 0.81          & \textbf{0.83} \\ 
YEAST       & Biology   & 1,484     & 0.00 & 10& \colorbox{orange}{0.46} & \colorbox{orange}{0.53}          & 0.54                                & 0.47          & \textbf{0.55}  \\
20NEWS      & Text      & 19,938    & 0.00 & 20& \colorbox{orange}{0.20} & \colorbox{orange}{0.49}          & 0.62                                & 0.16          & \textbf{0.63} \\ 
MED         & Text      & 1,033     & 0.00 & 31& \colorbox{orange}{0.50} & \colorbox{orange}{0.54}          & 0.54                                & 0.48          & \textbf{0.56} \\ 
REUTERS     & Text      & 8,293     & 0.00 & 65& \colorbox{orange}{0.60} & \colorbox{orange}{0.69}          & 0.75                                & 0.60          & \textbf{0.77} \\ 
ALPHADIGS   & Digit     & 1,404     & 0.00 & 6 & \colorbox{orange}{0.42} & \colorbox{orange}{0.48}          & 0.48                                & 0.46          & \textbf{0.51} \\ 
ORL         & Face      & 400       & 0.01 & 40& \colorbox{orange}{0.76} & \colorbox{orange}{0.82}          & 0.76                                & 0.78          & \textbf{0.83} \\ 
OPTDIGIT    & Digit     & 5,620     & 0.00 & 10& \colorbox{orange}{0.90} & \colorbox{orange}{0.93}          & 0.91                                & 0.90          & \textbf{0.98} \\ 
PIE         & Face      & 1,166     & 0.00 & 53& \colorbox{orange}{0.53} & \colorbox{orange}{0.66}          & 0.62                                & 0.51          & \textbf{0.74}\\ 
SEG         & Image     & 2,310     & 0.00 & 7 & \colorbox{orange}{0.54} & \colorbox{orange}{0.64}          & 0.59                                & 0.51          & \textbf{0.73} \\ 
UMIST       & Face      & 575       & 0.01 & 20& \colorbox{green}{0.74}  & \colorbox{green}{0.71}          & 0.67                                & 0.67          & \textbf{0.74} \\ 
PENDIGITS   & Digit     & 10,992    & 0.00 & 10& \colorbox{green}{0.82}  & \colorbox{green}{0.73}          & 0.82                                & 0.83          & \textbf{0.87}\\ 
SEMEION     & Digit     & 1,593     & 0.00 & 10& \colorbox{green}{0.86}  & \colorbox{green}{0.82}          & 0.77                                & 0.81          & \textbf{0.94} \\ 
AMLALL      & Medical   & 38        & 0.13 & 2 & \colorbox{orange}{0.92} & \colorbox{orange}{\textbf{0.95}} & 0.94                                & \textbf{0.95} & 0.92 \\ 
IONOSPHERE  & Radar     & 351       & 0.01 & 2 & \colorbox{yellow}{0.77} & \colorbox{yellow}{0.77}          & \textbf{0.85}                       & \textbf{0.85} & 0.70 \\ 

\multicolumn{8}{l}{\bf Undirected data} \\
POLBOOKS & Social & 105 & 0.08 & 3  & 0.83 & \textbf{0.85} & \textbf{0.85} & 0.82 & 0.83 \\ 
KOREA & Social &  35 & 0.11 & 2 & \textbf{1.00} & \textbf{1.00} & \textbf{1.00} & 0.71 & \textbf{1.00} \\ 
FOOTBALL & Sports & 115 & 0.09 & 12 & \textbf{0.94} & 0.93 & 0.90 & 0.93 & 0.93 \\ 
MIREX & Music & 3,090 & 0.00 & 10 &  0.21 & 0.24 & 0.27 & 0.12 & \textbf{0.43} \\ 
HIGHSCHOOL & Social & 60 & 0.10 & 5 & 0.82 & 0.85 & 0.83 & 0.82 & \textbf{0.95} \\ 
\bottomrule
\end{tabular}
\caption{Purity Measure Table}
\label{tab:purmeas}
\end{table}
\end{footnotesize}

We have discussed the output of a variety of experiments on a large number of data sets, but we also want to discuss their dependence upon the $\epsilon$ parameter and the percentage of nodes that are learned in the energy \eqref{eqn:parEssl}.  To that end, we consider the output purity measure for some representative data sets and look at the outputs over a range of epsilon parameters and percentages of learning.  In this case, we considered only the $K$-means initialization for consistency and simplicity of comparison.  For the $\epsilon$ sweep, we recall that we considered the range $\epsilon = {50 \nu}/{\| L \|_{{\rm Fro}}}$, where $\nu = e^{.2 \ell}$ with $-50 < \ell < 50$. In \cref{fig:epsweep} we show the variation in the purity measure with $\epsilon$ for a small graph ({\rm FOOTBALL}), a medium sized graph ({\rm OPTDIGITS)}, and a large graph ({\rm SEISMIC}). Similarly, in \cref{fig:persweep} we visualize how results vary with the fraction of supervision (nodes with labels provided) under semi-supervised learning, for the same graphs, with $\nu = .6,.8,1.0,1.2,1.4,1.6,1.8$.

\begin{figure}[t!]
\centering
\begin{subfigure}{.25\textwidth}
\includegraphics[width=\textwidth]{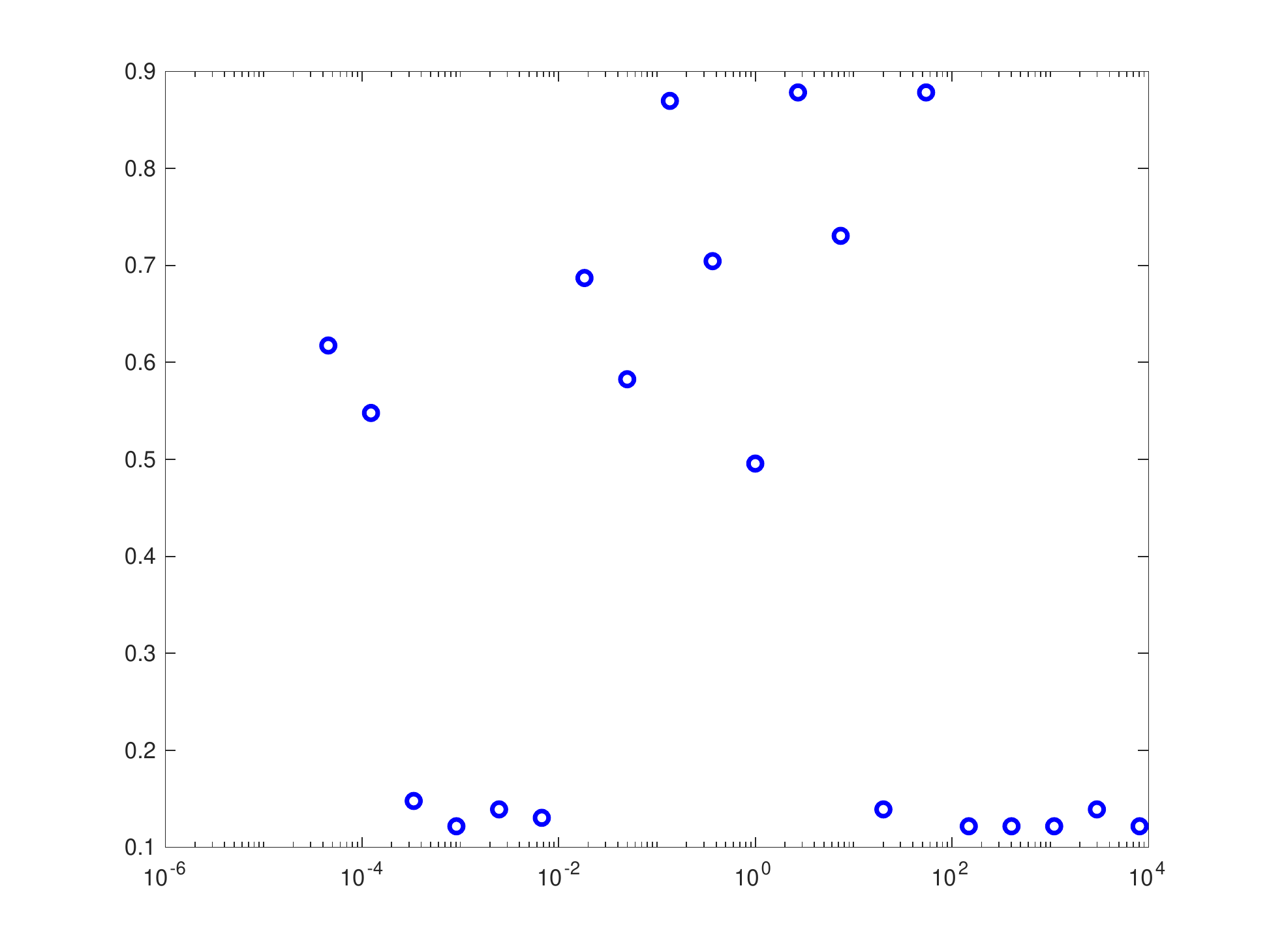}
\caption{Football}
\end{subfigure}
\begin{subfigure}{.25\textwidth}
\includegraphics[width=\textwidth]{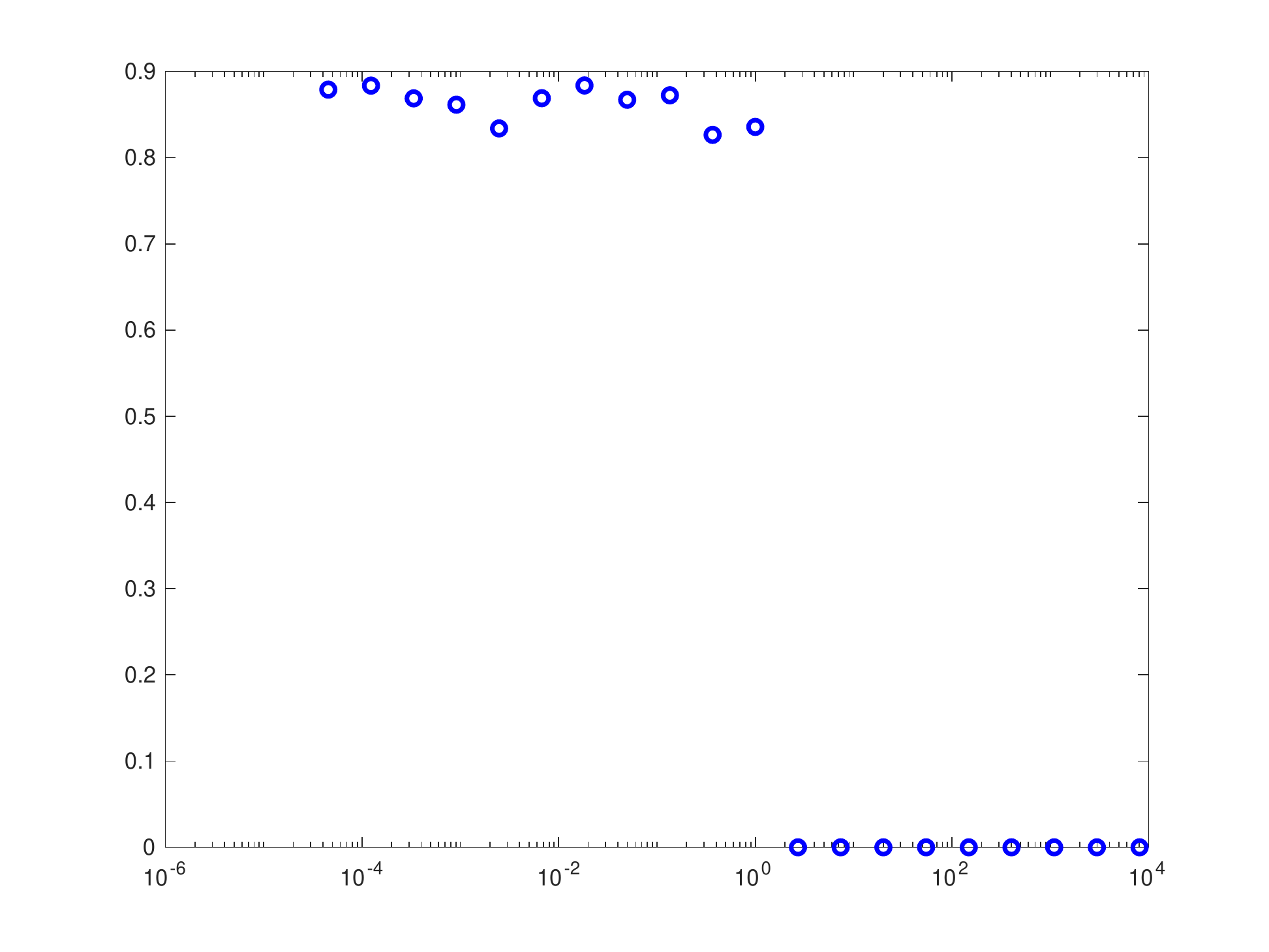}
\caption{Optdigits}
\end{subfigure}
\begin{subfigure}{.25\textwidth}
\includegraphics[width=\textwidth]{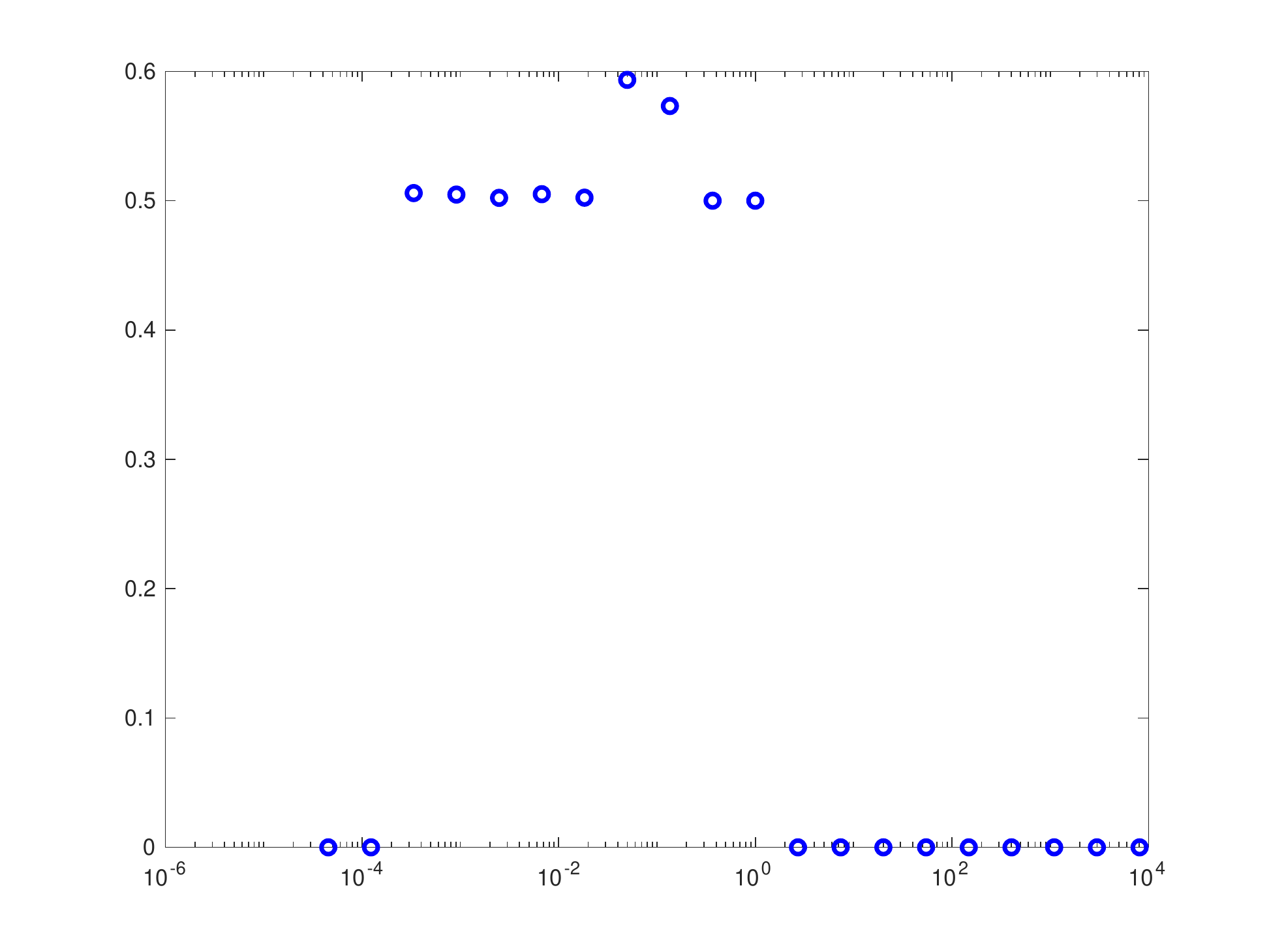}
\caption{Seismic}
\end{subfigure}
\caption{Purity measures for three selected data sets as a function of the scale parameter $\nu$. In all three panels, we observe a stable range (on a log scale) where purity is stably nontrivial, and in the left panel, there are two such scales.} 
\label{fig:epsweep}
\end{figure}

\begin{figure}[t!]
\centering
\begin{subfigure}{.25\textwidth}
\includegraphics[width=\textwidth]{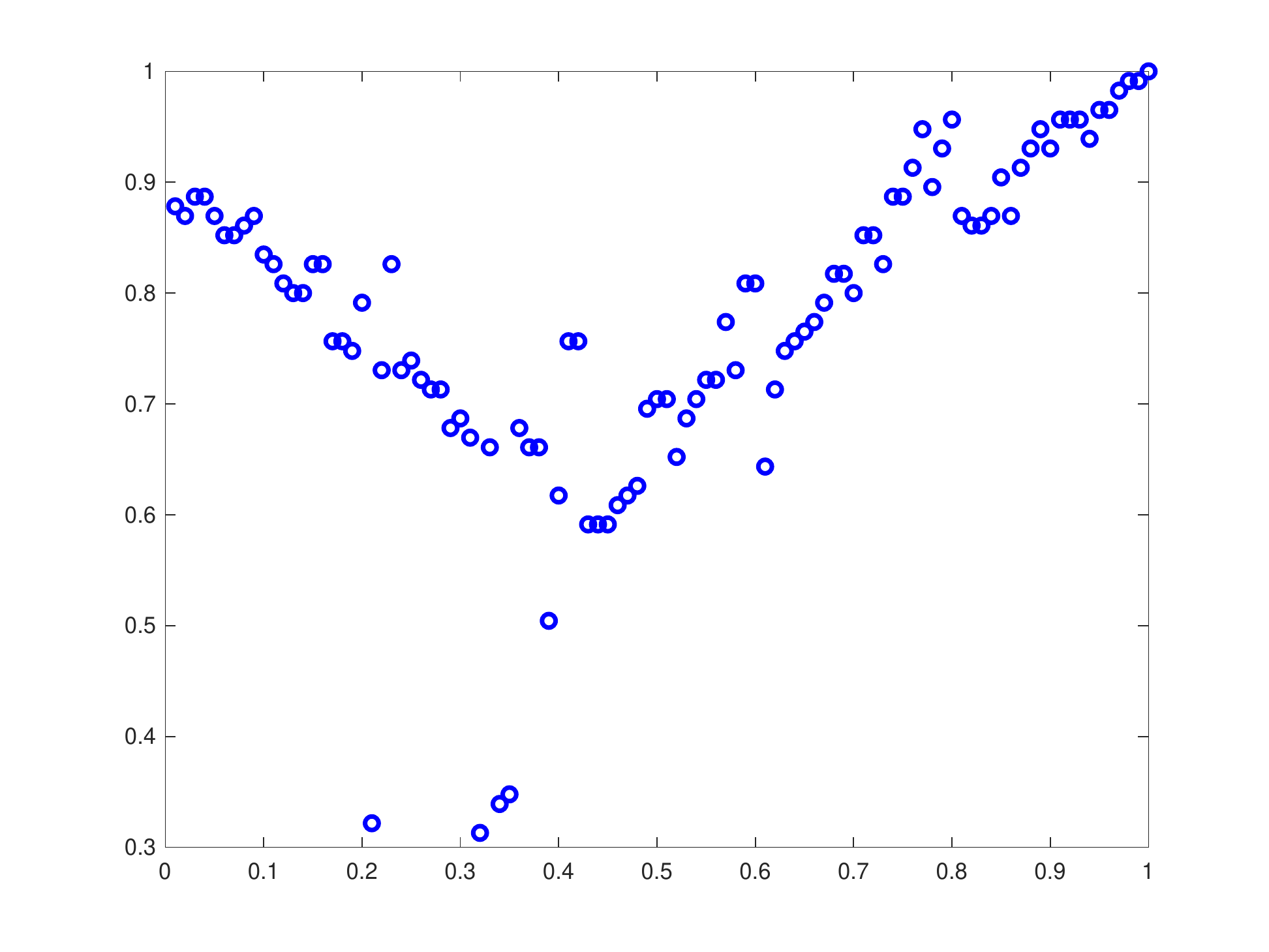}
\caption{Football}
\end{subfigure}
\begin{subfigure}{.25\textwidth}
\includegraphics[width=\textwidth]{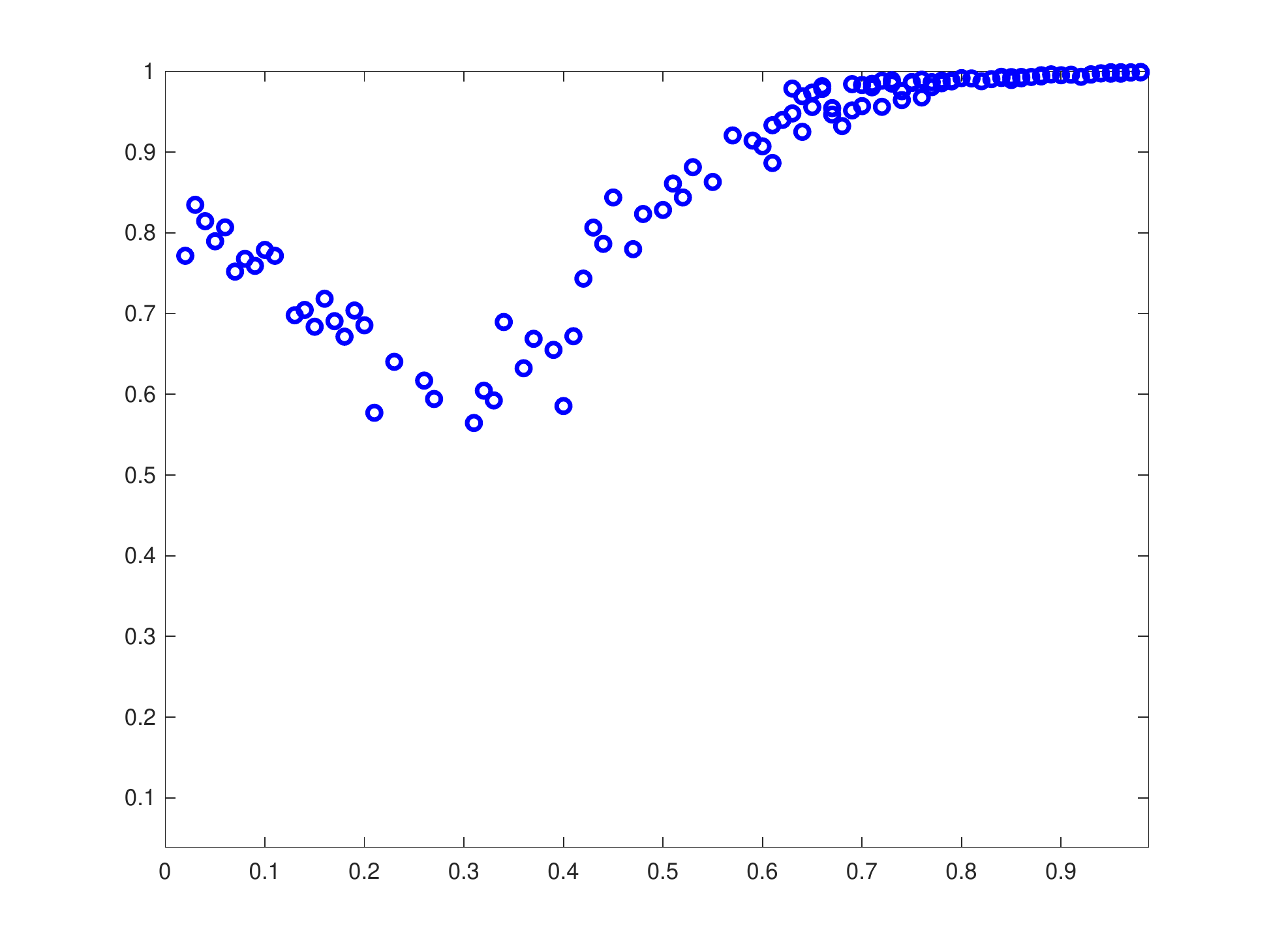}
\caption{Optdigits}
\end{subfigure}
\begin{subfigure}{.25\textwidth}
\includegraphics[width=\textwidth]{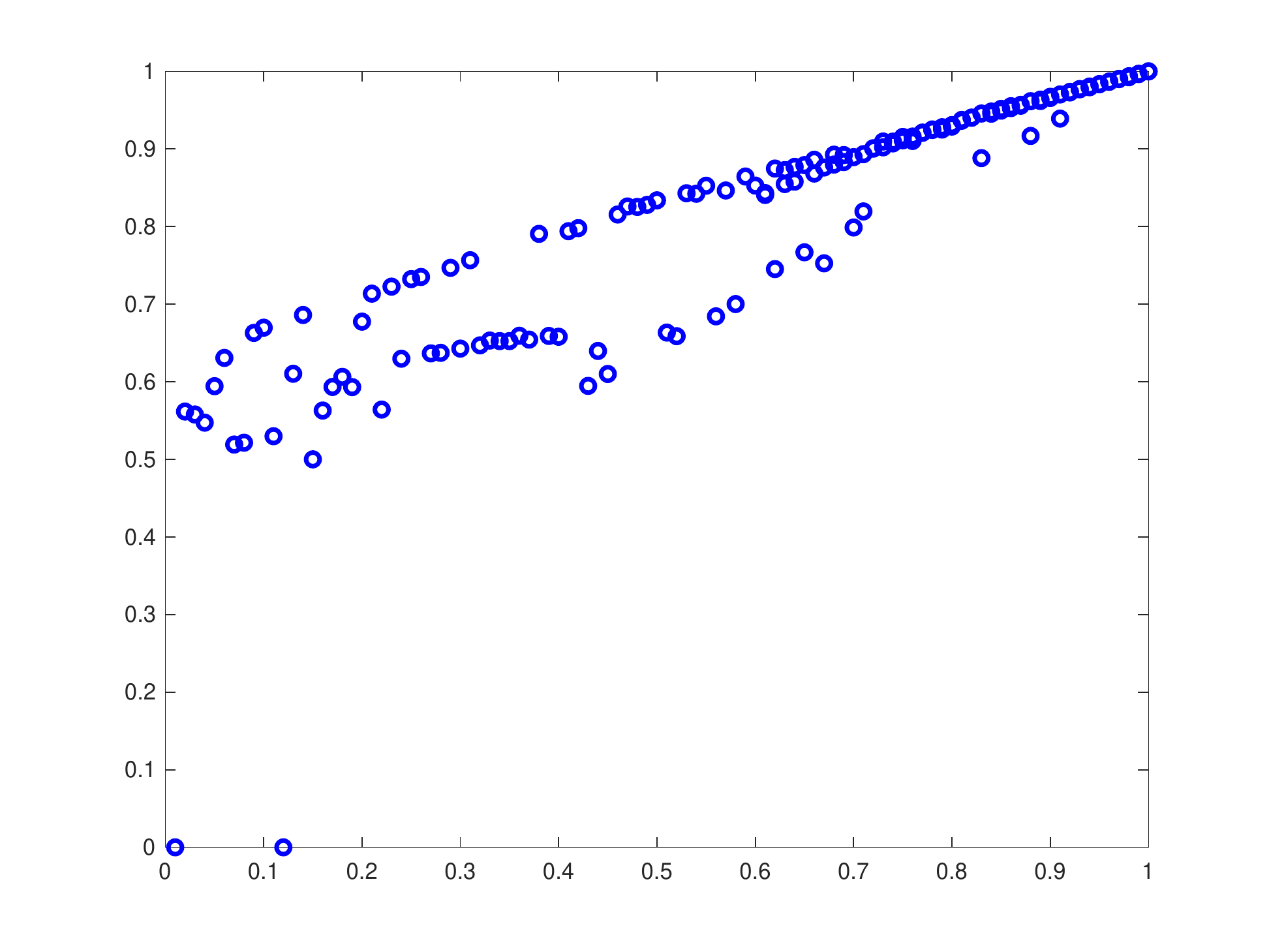}
\caption{Seismic}
\end{subfigure}
\caption{Purity measures on three selected data sets as a function of the fraction of supervision (nodes with labels provided) under semi-supervised learning. We observe that supervision can either consistently help (as in the right panel) or can have inconsistent effects (as in the left and middle panels). Once possible explanation for this is that there may be multiple clustering structures present in the data, and it takes a lot of supervision to force the partitioner to switch to a partition aligned with the metadata indicated by the supervision, rather than a different clustering structure that is better from the perspective of the optimizer.}
\label{fig:persweep}
\end{figure}

\section{Discussion}
\label{s:disc}

Throughout our study we emphasize that our methodology operates fundamentally on the possibly directed nature of the underlying graph data. Considering the Index of Complex Networks \cite{ICON} as a representative collection of widely-studied networks, we note that (as of our writing here) 327 of the 698 entries in the Index contain directed data. Whereas there are undoubtedly settings where one can ignore edge direction, there are inevitably others where respecting direction is essential. By formulating a strategy for subgraph detection and graph partitioning inherently built on processes running on the directed graph, we avoid the need for any \emph{post hoc} modifications to try to respect directed edges. In particular, our method nowhere relies on any correspondingly undirected version of the graph, avoiding possible information lost in symmetrizing. 

While we expect that our formulation of escape times can be useful in general, including for undirected graphs, our proper treatment of the directed graph data should prove especially useful. For example, the directed follower v.\ following nature of some online social networks (e.g., Twitter) is undoubtedly important for understanding the processes involved in the viral spread of (mis)information. As shown by \cite{Weng_Menczer_Ahn_2013} (and extended by \cite{li2019infectivity}), the community structure is particularly important for identifying the virality of memes specifically because a meme that ``escapes" (in our present language) its subgraph of origin is typically more likely to continue to propagate. Another application where directed escape times could be relevant is in detecting the (hidden) circulation of information, currency, and resources that is part of coordinated adversarial activity, as explored for example in~\cite{jin2019noisy,moorman2018filtering,sussman2020matched}.


To close, we highlight two related thematic areas for possible future work that we believe would lead to important extensions on the methods presented here. 

\subsection{Connection to distances on directed graphs}
In previous work of the present authors \cite{boyd2020metric} along with Jonathan Weare, we construct a symmetrized distance function on the vertices of a directed graph.  We recall the details briefly here, which is based somewhat upon the hitting probability matrix construction used in umbrella sampling (\cite{dinner2017stratification,Thiede_2015}).  For a general probability transition matrix $P$, we denote the Perron eigenvector as
\[
P' \phi = \phi.
\]
Let us define a matrix $M$ such that $M_{ij} = \prob_i [\tau_j < \tau_i]$, where $\prob_i [\tau_j < \tau_i]$ is the probability that starting from site $i$ the hitting time of $j$ is less than the time it takes to return to $i$. Let $X(t)$ be the Markov chain with transition matrix $P$. Then, it can be observed that (\cite{dinner2017stratification,Thiede_2015})
\[
\prob_i [\tau_j < \tau_i] \phi_i = \prob_j [\tau_i < \tau_j] \phi_j,
\]
where $\prob_i$ represents the probability from $X(0) = i$. 
This means that from hitting times, one can construct a symmetric adjacency matrix,
\begin{equation}
\label{Aht}
A^{(hp)}_{ij} = \frac{ \sqrt{\phi_i} }{ \sqrt{\phi_j} } \prob_i [\tau_j < \tau_i] = A^{(hp)}_{ji}\,.
\end{equation}
This adjacency matrix has built-in edge weights based upon hitting times, and we can then easily partition this adjacency matrix using our symmetric algorithms, in particular the mean exit time fast algorithm developed here.  
The distance function in \cite{boyd2020metric} is given by $d^\beta \colon [n] \times [n] \to \mathbb R$, which we refer to as the \emph{hitting probability pseudo-metric}, by 
\begin{equation} \label{e:Dist}
  d (i,j) = - \log \left( A^{(hp)}_{ij}  \right). 
\end{equation}
This is generically a pseudo-metric as it is possible distinct nodes can be distance $0$ from one another, however there exists a quotient graph on which $d$ is a genuine metric.
Indeed, a family of metrics is given in \cite{boyd2020metric} that has to do with possible choices of the normalization in \cref{Aht} with different powers of the invariant measure.  
A natural question to pursue is whether parsing the directed network with this approach to create the symmetrized $A^{(hp)}$ matrix, then applying our clustering scheme can be used to effectively detect graph structures in a more robust manner.  In particular, comparison of our clustering scheme versus $K$-means studies of the distance structure should be an important direction for future study.

\subsection{Continuum Limits}
The methods presented here have a clear analog in the continuum setting to the motivated problems in the continuum discussed in the introduction.  The primary continuum problem is related to the landscape function, or torsion function, on a sub-domain prescribed with Dirichlet boundary conditions,
\begin{align}
-\Delta u_S = 1_S, \ \ u_S |_{\partial S} = 0.  
\end{align}
This is known as the mean exit time from a set $S$ of a standard Brownian motion random walker, see \cite{pavliotis2014stochastic}, Chapter $7$.  Correspondingly, for a domain $\Omega$ with Neumann boundary conditions (to make life easier with graphs) and some $0 < \alpha < 1$, we propose the following optimization
\begin{align}
\max_{S \subset \Omega, |S| = \alpha |\Omega|} \int_S u_S\, dx\,,
\end{align}
meaning that we wish to maximize the exit time of a random walker from a given sub-domain.  Through the Poisson formula for the mean exit time, we have that $\int u_S = (- \Delta u_S, u_S)$, allowing us to frame things similarly via a Ginzburg--Landau like penalty term for being in a set $S$,
$$
\min_{ \substack{0\leq \phi \leq 1 \\ \int \phi = \alpha |\Omega| }} \  \min_{ \int u = 1 } \frac12 (- \Delta u_S, u_S) + \frac{1}{2 \epsilon} \langle u, (1-\phi) u \rangle.
$$
Analysis of optimizers for such a continuum problem and its use in finding sub-domains and domain partitions is one important direction for future study. 
Related results in a continuum setting have been studied for instance in \cite{briancon2004regularity,buttazzo1993existence}, but the regularization of this problem seems to be new and connects the problem through the inverse of the Laplacian to the full domain and its boundary conditions.  Following works such as \cite{osting2017consistency,singer2017spectral,trillos2016continuum,trillos2018variational,trillos2016consistency,YUAN_2021}, an interesting future direction would be to prove consistency of our algorithm to these well-posed continuum optimization problems.  


\bibliographystyle{amsplain}
\bibliography{refs}

\end{document}